\documentclass[reqno,11pt]{amsart}
\usepackage{amsmath,amsthm,amscd,amssymb,latexsym,upref,stmaryrd}
\usepackage{amsfonts,mathrsfs,dsfont}
\usepackage[top=1.5in, bottom=1.5in, left=1.4in, right=1.4in]{geometry}
\usepackage{color}

\newcommand{\R}{{\bbR}}

\newcommand{\C}{{\mathbb C}}

\newcommand{\IR}{\mathbb{R}}
\newcommand{\ID}{\mathbb{D}}

\newcommand{\bbD}{{\mathbb{D}}}

\newcommand{\bbP}{{\mathbb{P}}}

\newcommand{\bbR}{{\mathbb{R}}}

\newcommand{\bsX}{{\boldsymbol{X}}}

\newcommand{\cC}{{\mathcal C}}
\newcommand{\cD}{{\mathcal D}}

\newcommand{\cF}{{\mathcal F}}

\newcommand{\cH}{{\mathcal H}}

\newcommand{\cO}{{\mathcal O}}

\renewcommand{\d}{{\delta}}

\renewcommand{\O}{{\Omega}}

\DeclareMathOperator{\const}{const}

\DeclareMathOperator{\dist}{dist}
\DeclareMathOperator{\supp}{supp}

\renewcommand{\ln}{\text{\rm ln}}

\newcommand{\beq}{\begin{equation}}
\newcommand{\enq}{\end{equation}}

\let\geq\geqslant
\let\leq\leqslant

\makeatletter
\def\theequation{\@arabic\c@equation}

\allowdisplaybreaks 
\numberwithin{equation}{section}

\newtheorem{theorem}{Theorem}[section]
\newtheorem{proposition}[theorem]{Proposition}
\newtheorem{lemma}[theorem]{Lemma}
\newtheorem{corollary}[theorem]{Corollary}

\theoremstyle{remark}
\newtheorem{remark}[theorem]{Remark}

\begin{document}

\title[Weak anisotropic Hardy inequality ]{Weak anisotropic Hardy inequality: essential self-adjointness of 
drift-diffusion operators on domains in $\IR^d$, revisited}

\author[G.\ Nenciu]{Gheorghe Nenciu}
\address{Gheorghe Nenciu\\Institute of Mathematics ``Simion Stoilow'' of the Romanian Academy\\ 21, Calea Grivi\c tei\\010702-Bucharest, Sector 1\\Romania}
\email{Gheorghe.Nenciu@imar.ro}

\author[I.\ Nenciu]{Irina Nenciu}
\address{Irina Nenciu\\
         Department of Mathematics, Statistics and Computer Science\\ 
         University of Illinois at Chicago\\         851 S. Morgan Street\\
         Chicago, IL \textit{and} Institute of Mathematics ``Simion Stoilow''
     of the Romanian Academy\\ 21, Calea Grivi\c tei\\010702-Bucharest, Sector 1\\Romania}
\email{nenciu@uic.edu}

\begin{abstract}
We consider the problem of essential self-adjointness of the drift-diffusion operator 
$H=-\frac{1}{\rho}\nabla\cdot \rho \mathbb D\nabla +V$  on  domains $\Omega  \subset \bbR^d$
with $\cC^2$-boundary $\partial \Omega$ and for large classes of coefficients  $\rho,\; \bbD$ and $V$.
We give criteria showing how the behavior as $x \rightarrow \partial \Omega$ of these coefficients balances to ensure 
essential self-adjointness of $H$. On the way we prove a weak anisotropic Hardy inequality which is of
independent interest.
\end{abstract}

\maketitle

\tableofcontents

\section{Introduction}\label{S:1}

In this paper we revisit the problem of essential self-adjointness for a drift-diffusion operator on domains in $\bbR^d$. 
More precisely, we consider the operator
\begin{equation}
H=-\frac{1}{\rho}\nabla\cdot \rho \mathbb D\nabla +V
\end{equation}
on a domain $\Omega \in \bbR^d$, where  $ \rho :\Omega \rightarrow (0, \infty)$, 
$\bbD: \Omega \rightarrow \bbR^{d \times d}$ is a  strictly positive definite matrix valued
function, and $V: \Omega  \rightarrow \bbR$.

The self-adjointness problem for second order elliptic partial differential operators
on domains in $\bbR^d$ (or more general on Riemannian manifolds) has a long and ramified history.
In the 1-dimensional case, the problem is well understood due to the powerful tools of ordinary differential equations theory, 
especially Weyl limit point/circle theory \cite{RS}, \cite{RoS}. The multi-dimensional case has been also much studied 
and we send the reader to \cite{BMS}, \cite{Br}, \cite{Co}, \cite{KSWW}, \cite{RS} for references, heuristics, and a detailed account.  
For more recent developments see \cite{BL}, \cite{BP}, \cite{CT}, \cite{MT}, \cite{NN1}, \cite{NN2}, \cite{PRS},
\cite{Ro1}, \cite{Ro2}.

Assuming enough regularity of the coefficients, $H$ is symmetric in $L^2\big(\Omega; \rho\,dx\big)$
on $\cC^{\infty}_0(\Omega)$ and as it is well known (see e.g. \cite{CT}) the essential self-adjointness
of $H$ is dictated only by their behavior near the boundary, $\partial \Omega$ of $\Omega$.
Hence the problem boils down to finding conditions on the behavior of its coefficients as $x \rightarrow \partial \Omega$ 
ensuring essential self-adjointness of $H$. 
The case which is well understood is when $\bbD$ is strongly degenerate, i.e. it decays sufficiently quickly as 
$ x \rightarrow \partial \Omega$ so that $\Omega$ endowed with the distance
\begin{equation}
ds^2 =\sum_{j,k=1}^d \bbD(x)^{-1}_{j,k}dx_j dx_k
\end{equation} 
is a complete metric space  (e.g.  in the case when $\|\bbD(x)\| \leq \const.\delta(x)^2$
where $\delta (x) = \dist \big(x, \partial \Omega\big)$). One particular result in this case, which follows from the general  theory in 
\cite{BMS}, is that $H_0$ is essentially self-adjoint irrespective of the geometry of $\Omega$ and the behavior of $\rho$. 
The situation when $\bbD$ is not strongly degenerate is much more involved and after more than five decades of research (see e.g. 
the line of research following the famous paper of Wienholtz \cite{S-H}, \cite{W}, \cite{Wie}, the Cordes approach \cite{Co} and the 
Brusentev theory \cite{Br}) the picture is far from being complete. More recently, for the 
particular case of Schr\"odinger operator, $H=-\Delta + V$ on bounded, simply connected domains  with smooth boundary we refined 
the previously known results on the behavior
of $V(x)$  as  $ x \rightarrow \partial \Omega$ ensuring essential self-adjointness \cite{NN1}. The key point of the method in \cite{NN1} 
is that essential self-adjointness is ensured by lower bounds of
$\big(\varphi, H\varphi\big),\; \varphi \in \cC^{\infty}_0$:
\begin{equation}\label{I:lb} 
\big(\varphi, \big(-\Delta + V\big)\varphi\big) \geq \int_{\Omega} \big(L(x) +V(x)\big) \big| \varphi
(x)\big|^2dx
\end{equation}
via Agmon-type exponential estimates for weak solutions of partial differential equations \cite{Ag}. 
The heuristics behind is that if the "barrier" $L(x) +V(x) \rightarrow \infty$
sufficiently fast as $x \rightarrow \partial \Omega$ the quantum particle cannot reach $\partial \Omega$, 
hence no boundary conditions are needed to specify the dynamics. It turns out that 
the method in \cite{NN1} works as well in more general settings and in particular,  very recently,
criteria for  essential self-adjointness of some classes of drift-diffusion operators were obtained
in \cite{NN2} and (by merging the approaches in \cite{NN2}, \cite{Ro1}) in \cite{Ro2}.

From  the above discussion it is clear that the main ingredient  which is needed in this approach is a (as "optimal" as possible) lower 
bound (aka Hardy inequality) for the general drift-diffusion operator $H_0=-\frac{1}{\rho}\nabla\cdot \rho \mathbb D\nabla$
\begin{equation}\label{I:H0lb} 
\big(\varphi, H_0\varphi\big) \geq \int_{\Omega} L(x)  \big| \varphi
(x)\big|^2 \rho(x)dx.
\end{equation}

If $\rho =1,\; \bbD=\mathds{1}$ then \eqref{I:H0lb} is provided by the well known weak Hardy inequality (see e.g. \cite{BM}, \cite{Da})
\begin{equation}
L(x)= C +\frac{1}{4 \delta(x)^2}, \; C> -\infty,\; \delta(x)= \dist(x, \partial \Omega).
\end{equation}
More generally if $\bbD(x)$ is isotropic i.e. $\bbD(x) = d(x)\mathds{1}= a(x)\delta(x)^{\beta}\mathds{1}$ then $L(x)$ is provided by weighted Hardy inequalities (see  \cite{Ro3} and references given there). Weighted Hardy inequalities can be used also for anisotropic $\bbD(x)$ by noticing that
\begin{equation}\label{I:H0min}
\big(\varphi, H_0\varphi\big) \geq \int_{\Omega} \big|\nabla \varphi(x)\big|^2 d(x)\rho(x)dx
\end{equation}
if $d(x)$ satisfies $ \bbD(x) \geq d(x)\mathds{1}$ but \eqref{I:H0min} might be very poor if $\bbD(x)$ is not in a "neighborhood" of an isotropic one.

The aim of this paper is to extend the results in \cite{NN2}, \cite{Ro2} to larger classes of $\rho,\; \bbD$ and $V$ in order   to allow the degree of degeneracy of $\bbD$ to vary with the point on $\partial \Omega$ and to dispense of isotropy like conditions. On the way we prove a weak anisotropic Hardy inequality which might be of independent interest.

The content of the paper is as follows. In Section \ref{S:2} we spell the setting for Sections
\ref{S:3} and \ref{S:4}. Section \ref{S:3} is devoted to the proof of our weak anisotropic  Hardy inequality. In section \ref{S:4} combining the result in Section \ref{S:3} with the method in \cite{NN1}, \cite{NN2} we provide criteria  for essential self-adjointness oh $H$ for  classes
 of coefficients described in Section \ref{S:2}. While in order not to obscure the main ideas in
 Sections \ref{S:3} and \ref{S:4} we restricted ourselves to the setting in Section \ref{S:2}, the method of proof works as well (adding the necessary technicalities) in more general settings, and
 we give in Section \ref{S:5} some refinements and extensions of the results in Sections \ref{S:3} and \ref{S:4}. Comments about relation with previous results, optimality, as well as an application to the essential self-adjointness of Laplace-Beltrami operators on two dimensional almost-Riemannian structures \cite{BBP}, \cite{BL} are also given.

\section{The setting}\label{S:2}
Through the paper, if not otherwise stated, $\Omega$ denotes a bounded simple connected domain in $\R^d$
whose boundary $\partial \Omega$ is a $\cC^2$ manifold of codimension 1. Let 

\begin{equation}\label{S:delta}
\delta (x) = \dist \big(x, \partial \Omega\big); \quad x \in \O
\end{equation}
be the distance to the boundary and 

\begin{equation}\label{S:gammanu}
\Gamma_{\nu} =\big\{ x \in \Omega | \d(x) < \nu \big\} ; \quad \nu >0.
\end{equation}

The main properties of the function $\delta(x)$ we need are given by the following:

\begin{lemma}\label{GT}\cite[ Ch. 14.6]{GT}
There exist $\nu_{\Omega} >0 $ and $C_{\Delta} < \infty$ such that
\begin{equation}\label{S:graddelta}
\delta \in \cC^2(\Gamma_{\nu_{\Omega}});\quad \big|\nabla\delta(x)\big|=1; \quad \big|\Delta \delta(x) \big|
\leq C_{\Delta}.
\end{equation}
\end{lemma}

Now let
\begin{equation}\label{S:rho}
\rho \in \cC^1\big(\Omega; (0,\infty)\big),
\end{equation}
$\bbD(x)$ a matrix valued function:
\begin{equation}\label{S:D}
\bbD \in \cC^1\big(\Omega; \bbR^{d\times d}\big); \quad \bbD(x) >0
\end{equation}
and 
\begin{equation}\label{S:V}
V \in L^{\infty}_{loc}\big(\Omega; \bbR \big).
\end{equation}
With these we consider the following symmetric operator in $L^2\big(\Omega; \rho(x)dx\big)$:

\begin{equation}\label{S:H}
H=-\frac{1}{\rho} \nabla\cdot \rho \mathbb D\nabla +V; \quad \cD (H)= \cC^{\infty}_0(\Omega).
\end{equation}

By $H_0$ we denote the above operator with $V=0$ i.e.
\begin{equation}\label{S:H0}
H_0=-\frac{1}{\rho} \nabla\cdot \rho\mathbb D\nabla .
\end{equation}

The corresponding quadratic forms will be denoted by  $h\big[\cdot,\cdot\big]$ and $h_0\big[\cdot,\cdot\big]$ respectively.

Next we spell the assumptions on $\rho$ and $\bbD$ which will be used through the paper:

\textbf{Assumption A} 

Let 
\begin{equation}\label{S:snu0}
s_{\beta} \leq s <1; \quad 0 <\nu_0 < \text{min}\big\{\nu_{\O}, e^{-1}\big\}.
\end{equation}
{$\mathbf {A_{ \rho}}$ }:

For $x \in \Gamma_{\nu_0}$,
\begin{equation}\label{S:Arho}
\rho(x) = r(x) \delta(x)^{\gamma(x)}; \quad r \in \cC^1\big(\Gamma_{\nu_0}; (0,\infty)\big);  \quad \gamma \in 
\cC^1\big(\Gamma_{\nu_0}; \bbR\big),
 \end{equation}
 
 \begin{equation}\label{S:Arhosup}
\sup_{x\in \Gamma_{\nu_0}} \big(\big| \nabla \ln  r(x)\big| +\big| \nabla \gamma(x)\big| \big)\delta(x)^s <\infty.
\end{equation}

{$\mathbf {A_D}$}:

For $x \in \Gamma_{\nu_0}$ let $\bbP_x$ be the orthogonal projection in $\bbR^d$ along $\nabla \delta(x)$.
Write $\bbD(x)$ as a $2 \times 2$ block 
matrix according to the decomposition $\bbR^d=\bbP_x\bbR^d \oplus (1- \bbP_x)\bbR_d$:
\begin{align}\label{S:blockD}
\bbD(x)&= \bbP_x \bbD(x) \bbP_x  +\bbP_x \bbD(x)(1- \bbP_x  )+
(1-\bbP_x)\bbD(x)\bbP_x + \bbP_x \bbD(x) \bbP_x \nonumber \\ 
&=:
 \left(\begin{matrix} d(x) &d_{12}(x)\\ d_{21}(x) & d_{22}(x)  \end{matrix}\right).
\end{align}

Assume that
\begin{equation}\label{S:AD}
d(x)=a(x)\delta(x)^{\beta(x)}; \quad  a \in \cC^1\big(\Gamma_{\nu_0}; (0,\infty)\big);  \quad \beta \in 
\cC^1\big(\Gamma_{\nu_0}; \bbR\big),
\end{equation}
\begin{equation}\label{S:ADsup}
\sup_{x\in \Gamma_{\nu_0}} \big| \nabla \ln  a(x)\big|\delta(x)^s< \infty;  \quad 
\sup_{x\in \Gamma_{\nu_0}}\big|\nabla \beta(x)\big|\delta(x)^{s_{\beta}} <\infty.
\end{equation}

For further use we notice a direct consequence of Assumption A.  More precisely since $\Omega$ is bounded, $\partial \Omega$ is smooth and $s_{\beta} \leq s <1$
from \eqref{S:Arhosup} and \eqref{S:ADsup} it follows that
\begin{equation}\label{S:cA}
\beta, \; \gamma, \; \ln a, \; \ln r \in L^{\infty}(\Gamma_{\nu_0}).
\end{equation}
In particular
\begin{equation}\label{S:aainvers}
a,\; \frac{1}{a}, \; r, \; \frac{1}{r}  \in L^{\infty}(\Gamma_{\nu_0}).
\end{equation}

Since $\bbD(x)>0$, $\bbD(x)^{-1}$ exists, $\bbD(x)^{-1} >0$ hence (see \eqref{S:blockD})
$d(x) >0,\; d_{22}(x)>0$. By Schur complement lemma \cite{HJ} (we use its form used by physicists, see e.g. Lemma 2.3 in \cite{JN}) $\bbD(x)^{-1}$ written as a $2\times2$ block matrix reads:
\begin{align}\label{H:blockDinvers}
&\bbD^{-1} =  \nonumber \\
&\left(\begin{matrix} \tilde d &-\tilde d d_{12}d_{22}^{-1}\\
 & \\
 -d_{22}^{-1}d_{21}\tilde d & d_{22}^{-1} +d_{22}^{-1}d_{21}\tilde d
d_{12}d_{22}^{-1}  \end{matrix}\right),
\end{align}
with 
\begin{equation}\label{H:dtilde}
\tilde d= d^{-1}\big(1-q\big)^{-1}
\end{equation}
where
\begin{equation}\label{H:q}
q =d^{-1}d_{12}d_{22}^{-1}d_{21}\,.
\end{equation}
Since $d(x),\; \tilde d(x), \;d_{22}^{-1}(x)$ are all strictly positive it follows that
\begin{equation}\label{H:q01}
0 \leq q(x)<1.
\end{equation}

\textbf{Assumption Q}

\begin{equation}\label{S:Q}
\sup_{x \in \Gamma_{\nu_0}}\big| \nabla \ln \big(1-q(x)\big)\big| \delta(x)^s <\infty.
\end{equation}

As in the case of $a, r$ from  \eqref{S:Q} one has
\begin{equation}\label{S:Qinfty}
\frac{1}{1-q} \in L^{\infty}\big(\Gamma_{\nu_0}\big).
\end{equation}

\section{Weak  anisotropic  Hardy inequality}\label{S:3}


\begin{theorem}\label{H}

Suppose that Assumptions A and  Q hold true. 
Then   there exists $0<\nu_1 <\frac{\nu_0}{2}$ such that  for all $0<\nu < \nu_1$ 
and  $\varphi
 \in \cC^{\infty}_0(\Omega)$:

\begin{align}\label{H:Hardy1} 
h_0\big[\varphi, \varphi\big]
=\int_{\Omega}\overline{\nabla\varphi(x)}\cdot \big(\bbD(x)\nabla\varphi(x)\big)\rho(x)dx\nonumber \\
\geq -c_1(\nu) \|\chi_{\Omega \setminus \Gamma{\nu}}\varphi \|^2+ \int_{\Gamma_{\nu}}\cH_0(x)|\varphi(x)|^2 \rho(x) dx
\end{align}
where $c_1(\nu) <\infty$,
$\chi_{\Omega \setminus \Gamma{\nu}}$ is the characteristic function of $\Omega \setminus \Gamma{\nu}$ and
\begin{equation}\label{H:calH1}
\cH_0(x)= \frac{1}{4}\big(1-q(x)\big)a(x)\delta(x)^{\beta(x)-2} \Big[ \big(\beta(x)+\gamma(x)-1\big)^2 +\frac{1}{2} 
\Big(\ln\tfrac{1}{\delta(x)}\Big)^{-2}\Big].
\end{equation}
In particular if $\varphi \in \cC^{\infty}_0(\Gamma_{\nu})$ then
\begin{equation}\label{H:HardyL}
h_0[\varphi,\varphi] \geq \int_{\Gamma_{\nu}}\cH_0(x)|\varphi(x)|^2 \rho(x) dx.
\end{equation} 
\end{theorem}

\begin{proof}
As in the proof  of Lemma 5.4 in \cite{NN2} we use the so called "vector field" approach to Hardy inequalities  see e.g.
\cite{Br}, \cite{BFT1}, \cite{M}, \cite{L} :

\begin{lemma}\cite[Theorem 4.1]{Br} \label{H:VFA}
Let $\bsX \in C^1\big(\Omega;\IR^d\big)$ be a differentiable real vector field on $\Omega$. Then
\begin{equation}\label{H:VFA}
h_0[\varphi, \varphi]  \geq  \int_{\Omega}\left(\nabla\cdot \bsX(x)-\bsX(x)\cdot \big(\rho (x) \ID(x)\big)^{-1}\bsX(x)\right)\,
\big|\varphi(x)\big|^2\,dx\,,
\end{equation}
for all $\varphi\in C_0^\infty(\Omega)$.
\end{lemma}

With this the proof consists in making an appropriate  Ansatz,  $\bsX_0(x)$,  for $\bsX(x)$ and estimating the r.h.s. of \eqref{H:VFA}.

Let for $x \in \Gamma_{\nu_0}$  (recall that $\nu_0 < \frac{1}{e}$ see \eqref{S:snu0}):
\begin{equation}\label{H:tildeX1}
\widetilde\bsX_0= \frac{1}{2}\big(1-q(x)\big)a(x)r(x)\delta(x)^{\beta(x)+\gamma(x)-1}
\big[\beta(x)+\gamma(x)-1 + f\big(\delta(x)\big)\big]\nabla\delta(x)
\end{equation}
with
\begin{equation}\label{H:f1}
f(t) = \frac{1}{\ln\frac{1}{t}}
\end{equation}
and then
\begin{equation}\label{H:X1}
\bsX_0(x)= \Psi(x)\widetilde\bsX_0(x)
\end{equation}
where $\Psi:\Omega \rightarrow [0,1], \; \Psi \in \cC^1(\Omega)$, and

\begin{equation}\label{H:Psi}
\Psi=
\begin{cases}
1 & \quad\text{for } x\in\Gamma_{\frac{\nu_0}{2}}\\
0 & \quad\text{for } x\not\in\Gamma_{\frac{3\nu_0}{4}} 
\end{cases}\,.
\end{equation}
By construction $\sup_{x\in\Omega}\big|\nabla \Psi(x)\big| <\infty$, hence from Assumption A for all $0<\nu < \frac{\nu_0}{2}$
\begin{equation}\label{H:sup}
\sup_{x\in \Omega\setminus \Gamma_\nu}\frac{1}{\rho(x)} \big|\nabla\cdot \bsX_0(x)-\bsX_0(x)\cdot \big(\rho(x) \bbD(x)\big)^{-1}\bsX_0(x)\big| =: M_{0,\nu} <\infty.
\end{equation}
On $\Gamma_{\nu}$, since $\nu <\frac{\nu_0}{2}$, from \eqref{H:X1} and \eqref{H:Psi} $\bsX_0(x)=\tilde\bsX_0(x)$ hence we have to compute $\nabla\cdot \widetilde\bsX_0(x)-\widetilde\bsX_0(x)\cdot \big(\rho(x) \bbD(x))^{-1}\widetilde\bsX_0(x)$.

By direct (a bit tedious) computation, from \eqref{H:tildeX1} one obtains for $x \in \Gamma_{\frac{\nu_0}{2}}$:
\begin{align}\label{H:divX1}
&\nabla \cdot \bsX_0(x)=\frac{1}{2}\big(1-q(x)\big)a(x)\rho(x)\delta(x)^{\beta(x)-2}\nonumber \\
&\Big[ \big(\beta(x)+
\gamma(x)-1\big)^2+\big(\beta(x)+
\gamma(x)-1\big)f\big(\delta(x)\big)
+\delta(x)f'\big(\delta(x)\big)+R_0(x) \Big]
\end{align}
with
\begin{align}\label{H:R1}
&R_0(x)=\frac{\delta(x)}{2}\Big\{\big[1+\big(\beta(x)+\gamma(x)-1+f\big(\delta(x)\big)
\big)\ln\delta(x)\big]\nonumber \\
&\big(\nabla \beta(x)+\nabla \gamma(x)\big)
+\big(\beta(x)+
\gamma(x)-1+\big(\delta(x)\big)\big)\big(\nabla \big(1-q(x)\big)a(x) r(x)\big)\Big\}\cdot\nabla\delta(x)\nonumber \\
&+\frac{\beta(x)+\gamma(x)-1+f\big(\delta(x)\big)}{2}\delta(x)\Delta\delta(x).
\end{align}
Taking into account  \eqref{S:graddelta}, \eqref{S:Arhosup}, \eqref{S:ADsup} and the fact that $\lim_{\delta(x)
\rightarrow 0}\delta(x)^{\frac{1-s}{2}} \ln \delta(x)
=0$ one obtains that uniformly for $s_{\beta}\leq s$
\begin{equation}\label{H:supR1}
\sup_{x\in \Gamma_{\frac{\nu_0}{2}}}\big|R_0(x)\big|\delta(x)^{\frac{s-1}{2}} =:R_{0,\nu_0} <\infty.
\end{equation}

We  compute now $\bsX_0(x)\cdot \big(\rho (x) \ID(x)\big)^{-1}\bsX_0(x))$. From \eqref{S:AD}, \eqref{H:blockDinvers}, \eqref{H:dtilde}, \eqref{H:tildeX1} and \eqref{H:X1},

\begin{align}\label{H:X1D-1X1}
&\bsX_0(x)\cdot \big(\rho (x) \ID(x)\big)^{-1}\bsX_0(x)) = 
\frac{1}{4}\big(1-q(x)\big)a(x)\rho(x)\delta(x)^{\beta(x)-2}\nonumber \\
&\Big[
\big(\beta(x)+\gamma(x)-1\big)^2 +2\big(\beta(x)+\gamma(x)-1\big)f\big(\delta(x)\big)
+f\big(\delta(x)\big)^2\Big].
\end{align}
Putting together \eqref{H:divX1} and \eqref{H:X1D-1X1} one obtains  
\begin{align}\label{H:divX1D-1X1}
&\nabla \cdot \bsX_0(x) -\bsX_0(x)\cdot \big(\rho (x) \ID(x)\big)^{-1}\bsX_0(x))=
\big(1-q(x)\big)a(x)\rho(x)\delta(x)^{\beta(x)-2}\nonumber \\
&\Big[\frac{1}{4}\big(\beta(x)+\gamma(x)-1\big)^2
+\frac{1}{2}\big(\delta(x)f'\big(\delta(x)\big)-\frac{1}{4}f\big(\delta(x)\big)^2 +R_0(x)\Big].
\end{align}
Further from \eqref{H:f1} 
\begin{equation}\label{H:f1M1}
tf'(t) -\frac{1}{2}f(t)^2=\frac{1}{2}\Big(\frac{1}{\ln \frac{1}{t}}\Big)^2.
\end{equation}
Finally from \eqref{H:divX1D-1X1} and \eqref{H:f1M1} (on $\Gamma_{\frac{\nu_0}{2}}$)
\begin{align}\label{H:divX1D-1X1R1}
&\nabla \cdot \bsX_0(x) -\bsX_0(x)\cdot \big(\rho (x) \ID(x)\big)^{-1}\bsX_0(x))=
\big(1-q(x)\big)a(x)\rho(x)\delta(x)^{\beta(x)-2}\nonumber \\
&\Big[\frac{1}{4}\big(\beta(x)+\gamma(x)-1\big)^2
+\frac{1}{4}\Big(\frac{1}{\ln \frac{1}{\delta(x)}}\Big)^2 +R_0(x)\Big].
\end{align}
From \eqref{H:supR1} there exists $0< \nu_1< \frac{\nu_0}{2}$ such that for $0<\nu <\nu_1$
\begin{equation}\label{H:M2R1}
\frac{1}{8}\Big(\frac{1}{\ln \frac{1}{\delta(x)}}\Big)^2 +R_0(x) \geq 0
\end {equation}
hence on $\Gamma_{\nu_1}$
\begin{align}\label{H:divX1D-1X1fin}
&\nabla \cdot \bsX_0(x) -\bsX_0(x)\cdot \big(\rho (x) \ID(x)\big)^{-1}\bsX_0(x))\geq \nonumber \\
&\frac{1}{4}\big(1-q(x)\big)a(x)\rho(x)\delta(x)^{\beta(x)-2}\Big[
\big(\beta(x)+\gamma(x)-1\big)^2
+ \frac{1}{2}\frac{1}{\Big(\ln \frac{1}{\delta(x)}\Big)^2}
\Big].
\end{align}
Then \eqref{H:Hardy1} with $c_1(\nu)=M_{0,\nu}$ and $\cH_0$  given  by  \eqref{H:calH1} follows from
Lemma \ref{H:VFA}, \eqref{H:sup} and \eqref{H:divX1D-1X1fin}.
\end{proof}


\section{Essential self-adjointness of drift-diffusion operators}\label{S:4}

In this section we give criteria for essential  self-adjointness of operator $H$ (see \eqref{S:H}) 
in terms of $\rho$, $\bbD$, $V$ in the setting given in Section \ref{S:2}. The strategy of the proof of these criteria is the one outlined in Section 3 in of \cite{NN2}. It relies on the following inequality.
\begin{lemma}\label{BI}\cite[Lemma 3.4]{NN2}
Assume that there exist $E_0>-\infty$ and a function $B\,:\,\Omega\,\rightarrow\,[0,\infty)$ such that all $\varphi\in C_0^\infty(\Omega)$, it holds that:
\begin{equation}\label{E:h}
( \varphi, (H-E_0)\varphi) \geq \int_\Omega |\varphi(x)|^2 B(x)\rho (x)\,dx\,.
\end{equation}
Let $\psi_E$ be a weak solution of $\big(H-E\big)\psi_E=0$ for some $E< E_0$,
and let $g$ be a real-valued, Lipschitz continuous function on $\Omega$ satisfying (a.e):
\begin{equation}\label{E:gB}
\nabla g(x)\cdot\big(\bbD(x)\nabla g(x)\big)
\leq  B(x)+\frac{|E-E_0|}{2}\,.
\end{equation}
Then
\begin{equation}\label{E:BI}
\big(\psi_E , f^2\psi_E\big)\leq \frac{2}{|E-E_0|} \big(\psi_E,|m|\psi_E\big)\,,
\end{equation}
where $f=e^g\phi$ with $\phi\,:\,\Omega\rightarrow[0,1]$, $\phi\in \cC^1_0(\Omega)$, and
\begin{align}\label{E:|m|}
& |m|(x)= e^{2g(x)}\Big\{2\big[\nabla g(x)\cdot \big(\bbD(x)\nabla g(x)\big)\big]^{\frac{1}{2}}\nonumber \\
&\big[\nabla \phi(x)\cdot \big(\bbD(x)\nabla \phi(x)\big)\big]^{\frac{1}{2}} +
\nabla \phi(x)\cdot \big(\bbD(x)\nabla \phi(x)\big)\Big\}.
\end{align}
\end{lemma}

The choice of $\phi$ to be used in Lemma \ref{BI} is as follows (see \cite{NN2}). Let for $l=1,2,...$
\begin{equation}\label{E:rl}
r_l=\frac{r_1}{2^{l-1}}\,,\quad r_1<\frac{\nu_0}{2},
\end{equation}
\begin{equation}\label{E:kl}
k_l(t)\in \cC^1\big((0,\infty)\big);\quad \supp k'_l \subset (r_{l+1},r_l) ;\quad
k_l(t)=
\begin{cases}
0\,, &\quad t>r_{l+1}\\
1\,, &\quad t<r_{l},
\end{cases}
\end{equation}
\begin{equation}\label{E:lkderiv}
\big|k'_l(t)\big| \leq \frac{2}{r_l-r_{j+l}}
\end{equation}
and set

\begin{equation}\label{E:phil}
\phi_l(x)=k_l(\delta(x)),
\end{equation}

\begin{equation}\label{E:Bl}
B_l=\big\{ x\in \Omega \big| r_{l+1} < \delta(x) <r_l\big\}.
\end{equation}
Let $|m_l|(x)$ as given by \eqref{E:|m|} with $f=e^{2g(x)}\phi_l(x)$. By construction $B_l$ are disjoint and
\begin{equation}\label{E:suppphil}
\supp |m_l| \subset \supp \nabla\phi_l \subset B_l.
\end{equation}

The next lemma gives the essential self-adjointness of $H$ providing the following estimation, which has to be verified for each particular case, holds true.

\begin{lemma}\label{E:supml}
Suppose
 that (uniformly in $l$)
\begin{equation}\label{E:suplm}
\sup_{x \in \Omega}|m_l|(x) \leq C_m l, \quad C_m <\infty.
\end{equation}
 Then $H$ is essentially self-adjoint.
 \end{lemma}
 \begin{proof}
 Since, by assumption (see Lemma \ref{BI}), $g(x)$ is continuous on $\Omega$ for any compact $K \subset \Omega$
 
\begin{equation}\label{E:Kinf}
\inf_{x\in K}e^{2g(x)} =c_K >0.
\end{equation} 
 Fix now arbitrarily a compact $K$ in $\Omega$. Since $\Omega \setminus \Gamma_{r_l}$ exhaust $\Omega$ as $l\rightarrow \infty$ there exists an integer $L(K)$ such that
$K \subset \big\{x\in \Omega \big| \delta (x) > r_{L(K)} \big\}$.
By construction (see \eqref{E:kl}, \eqref{E:phil})
 
\begin{equation}\label{E:phi1}
\phi_l\big|_K=1 \; \text{for}  \; l \geq L(K).
\end{equation}
 Then from \eqref{E:Kinf}, \eqref{E:phil}, \eqref{E:supml} and Lemma \eqref{BI} for all $l\geq L(K)$ 

\begin{equation}\label{E:BIlK}
 \int_{K}\big|\Psi_E(x)\big|^2 \rho(x) dx\leq \frac{1}{c_K} \int_{K}\big|\Psi_E(x)\big|^2 e^{2g(x)}\phi_l(x)^2\rho(x) dx \leq
\frac{2C_ml}{c_K(E_0-E)}\int_{B_l}\big|\Psi_E(x)\big|^2 \rho(x) dx.
\end{equation}
Summing \eqref{E:BIlK} over $l$ from $L(K)$ to $N$ and taking into account that $B_l$ are disjoint one  obtains 
\begin{equation}
\frac{c_K\big(E_0-E\big)}{2 C_m}\sum_{l=L(K)}^N\frac{1}{l}
 \int_{K}\big|\Psi_E(x)\big|^2 \rho(x) dx\leq  \|\Psi_E\|^2
\end{equation}
which in the limit $N\rightarrow \infty$ gives
\begin{equation}
 \int_{K}\big|\Psi_E(x)\big|^2 \rho(x) dx=0.
 \end{equation}
 Since $K$ was arbitrary, $\Psi_E =0$ and the proof is finished by invoking the basic criterion for 
 (essential) self-adjointness \cite[Proposition 3.8]{Sch}.
\end{proof}
  
  Summing up, for a given class of $\rho,\; \bbD,\;V$ the problem of essential self-adjointness of $H$ boils down to the problem of finding $g(x)$ satisfying the conditions of Lemma \ref{BI} and Lemma
\ref{E:supml} . 
  
As  a warming up we begin with the simple case when $\bbP_x\bbD\bbP_x$ is strongly degenerate at the boundary of $\Omega$ (i.e.
$\bbP_x\bbD\bbP_x \rightarrow 0$ quickly enough as $\delta(x) \rightarrow 0$) and $V$ is bounded from below.  
\begin{theorem}\label{ESAC}
Suppose that there exists $\mu < \frac{\nu_0}{2}$ such that
\begin{equation}\label{E:com}
\sup_{x\in \Gamma_{\mu}}\big\| \bbP_x\bbD(x)\bbP_x\big\| \delta^{-\beta(x)}  <\infty; \quad \inf_{x\in \Gamma_{\mu}}\beta(x)
\geq 2; \quad \inf_{x \in \Omega} V(x)>-\infty.
\end{equation}
Then $H$ is essentialy self-adjoint.
\end{theorem}
\begin{remark}
The fact that strong degeneracy of $\bbD$ ensure the essential self-adjointness of $H_0$ is well known. In particular the fact that if $\sup_{x\in \Gamma_{\mu}}\big\| \bbD(x)\big\| \delta^{-\beta(x)}  <\infty$ and $ \inf_{x\in \Gamma_{\mu}}\beta(x)
\geq 2$ then, without any conditions on $\rho$ and on regularity of $\partial \Omega$, $H_0$ is essentially self-adjoint 
follows from general results on essential self-adjointness of Laplace-Beltrami operators on complete weighted Riemannian manifolds (see e.g. \cite{BMS}, \cite{NN2}). \end{remark}
\begin{proof}
In this case the verification  of Lemma \ref{BI} and Lemma
\ref{E:supml} is immediate.  
 Without restricting the generality one can suppose that $V(x) \geq 0$ on $\Omega$. In this case we choose $g=0$ so that the condition\eqref{E:gB} is trivially satisfied with $E_0=0$ and $B(x) =0$.
The condition \eqref{E:h} is also trivially satisfied since $H_0 \geq 0,\; V\geq 0$.
It remains to verify Lemma \ref{E:supml}.  We take $r_1 =\mu$ in  \eqref{E:rl}. From \eqref{E:|m|}
(recall that $g=0$)
\begin{equation}
|m_l|(x) =\nabla \phi_l(x)\cdot \big(\bbD(x)\nabla \phi_l(x)\big)
\end{equation}
hence from  \eqref{E:lkderiv}, \eqref{E:phil} and \eqref{E:com}
 (recall that $\big|\nabla \delta (x)
\big|=1$)
one obtains that

\begin{equation}
\sup_{x \in \Omega}|m_l|(x) =:C_3 <\infty.
\end{equation}
and the proof is finished.
\end{proof}

We turn now to the  results concerning essential self-adjointness of $H$ when \\
$\inf_{x\in \Gamma_{\nu_0}}\beta(x) \geq 2$ does not hold true. Consider first the case 
when  $\nabla \beta(x)=0$.

\begin{theorem}\label{Mbetaconst} Suppose that
 Assumption A  and Assumption Q hold true, $\beta(x)= \beta <2$ and in addition there exists
 \begin{equation}\label{E:mu}
  0<\mu <\frac{\nu_0}{2}
 \end{equation}
 such that
  on $\Gamma_{\mu}$
\begin{equation}\label{E:Vbetaconst}
V(x)\geq 
 \delta(x)^{\beta-2}\Big[v(x)-w_{\mu}\delta(x)^\frac{1-s}{2}\Big];\; 0\leq w_{\mu} <\infty;\;
  v\in L^{\infty}(\Gamma_{\mu}; \bbR).
\end{equation}
Then:

i. If
\begin{equation}\label{E:ESA}
\inf_{\Gamma_{\mu}}\frac{ \big(1-q(x)\big(\beta+\gamma(x)-1\big)^2  +
\frac{4v(x)}{a(x)}}{\big(\beta-2\big)^2 } \geq 1
\end{equation}
$H$ is essentially self-adjoint.

ii. If
\begin{equation}\label{E:qsmall}
\sup_{x \in \Gamma_{\frac{\nu_0}{2}}}q(x)\delta(x)^{\frac{s-1}{2}}  =: Q_{\nu_0}<\infty
\end{equation}
and

 \begin{equation}\label{E:ESAC}
\inf_{\Gamma_{\mu}} \frac{\big(\beta(x)+\gamma(x)-1\big)^2  +
\frac{4v(x)}{a(x)}}{\big(\beta(x)-2\big)^2 } \geq 1,
\end{equation}
$H$ is essentially self-adjoint

\end{theorem}

\begin{proof}

i.  As in the proof of Theorem \ref{ESAC} one uses Lemma \ref{BI} but in this case one needs a nontrivial $g(x)$. The chioce of $g(x)$ is as follows:
\begin{equation}\label{E:gconst}
 g(x)=\Big(\frac{2-\beta}{2} \ln \delta(x)\Big)\Psi(x)
\end{equation}
where $\Psi$ (see \eqref{H:Psi}) is  the same as in the proof of Theorem \ref{H} .
From the fact that $g \in \cC^1(\Omega; \bbR^d)$,  $\bbD \in \cC(\Omega; \bbR^{d\times d})$
it follows that for all $\nu >0$
\begin{equation}\label{E:gnuconst}
\sup_{\Omega \setminus \Gamma_{\nu}} \nabla g(x)\cdot \big(\bbD(x)\nabla g(x)\big)=: g_{\nu} <\infty.
\end{equation} 

Next, we compute  $\nabla g(x)\cdot \big(\bbD(x)\nabla g(x)\big)$on $\Gamma_{\mu}$.
 For $\delta(x) <\frac{\nu_0}{2}$, $\Psi(x)=1$ hence for $x \in \Gamma_{\mu}$\begin{equation} \label{E:gradgconst}
\nabla g(x)= \frac{2-\beta}{2}\delta(x)^{-1} \nabla\delta(x).
\end{equation}
From \eqref{S:blockD} and \eqref{S:AD}
\begin{equation}
\nabla \delta(x)\cdot \big(\bbD(x)\nabla \delta(x)\big)=a(x)\delta(x)^{\beta},
\end{equation}
which together with \eqref{E:gradgconst} leads to
\begin{equation}\label{E:gradgMconst}
\nabla g(x)\cdot \big(\bbD(x)\nabla g(x)\big)=a(x)
\Big(\frac{2-\beta}{2}\Big)^2\delta(x)^{\beta-2} .
\end{equation}

The next step is to show that using Theorem \ref{H} and \eqref{E:gradgMconst} one can choose
$B(x);\; E_0,\; E$ such that \eqref{E:h} and \eqref{E:gB} hold true.
From \eqref{S:V} it follows that for all $\nu >0$
\begin{equation}\label{E:Vnu}
\text{ess-sup}_{x\in \Omega \setminus \Gamma_{\nu}} \big|V(x)\big| =: V_{\nu} < \infty
\end{equation}\label{E:infa}
and from \eqref{S:aainvers} for $0<\nu < \mu$
\begin{equation}\label{E:infa}
0<a_{\mu}:= \inf_{\Gamma_{\mu}}a(x) \leq \inf_{\Gamma_{\nu}}a(x).
\end{equation}

From  \eqref{E:Vbetaconst}, \eqref{E:Vnu}, \eqref{E:infa} and Theorem \ref{H}, for $0<\nu <\min\{\nu_1,\mu\}$
(see Theorem \ref{H} for $\nu_1$ and recall that on $\Gamma_{\mu}$, $\delta(x) <1$)
\begin{align}\label{E:hVbigger}
&h\big[\varphi,\varphi\big] \geq -\big( c_1(\nu) +
V_{\nu}\big)\| \varphi \|^2+     
\int_{\Gamma_{\nu}}a(x)  \delta(x)^{\beta -2}  
 \Big[\big(1-q(x)\big)\Big(\frac{\beta+ 
 \gamma (x) -1}{2}\Big)^2+ \nonumber \\ \frac{v(x)}{a(x)}+   
 & \frac{1-q(x)}{8}\frac{1}{
\ln \Big(\frac{1}{\delta(x)}\Big)^2}
 - \frac{w_{\mu}}{a_{\mu}}\delta(x)^{\frac{1-s}{2}}\Big] \big|\varphi (x)\big|^2 \rho(x)dx.
\end{align}

Since $\frac {1}{\ln \frac{1}{\delta(x)}} \rightarrow 0$ as $\delta(x)  \rightarrow 0$ slower than any
(strictly) positive power of $\delta(x)$ and $\frac{1}{1-q} \in L^{\infty}(\Gamma_{\mu})$ there exists 

\begin{equation}\label{E:nu2}
0<\nu_2 < \min\{\mu, \nu_1\}
\end{equation}
such  that on $\Gamma_{\nu_2}$  
\begin{equation}\label{E:lnconst}
\frac{1-q(x)}{8}\frac{1}{
\Big(\ln \frac{1}{\delta(x)}\Big)^2}
 - \frac{w_{\mu}}{a_{\mu}}\delta(x)^{\frac{1-s}{2}} \geq 0,
\end{equation}
hence on $\Gamma_{\nu_2}$, 
\begin{align}\label{E:hVbetaconst}
&h\big[\varphi,\varphi\big] \geq -\big( c_1(\nu_2) +
V_{\nu_2}\big)\| \varphi \|^2+     
\int_{\Gamma_{\nu_2}}a(x)  \delta(x)^{\beta -2}\nonumber \\  
 &\Big[\big(1-q(x)\big)\Big(\frac{\beta+ 
 \gamma (x) -1}{2}\Big)^2+  \frac{v(x)}{a(x)}\Big]   
  \big|\varphi (x)\big|^2 \rho(x)dx.
\end{align}
Now, from \eqref{E:hVbetaconst},  \eqref{E:h} holds true by choosing:
\begin{equation}\label{E:E0}
E_0 =-\big(c_1(\nu_2)+V_{\nu_2} \big)
\end{equation}
and 
\begin{equation}\label{E:Bbetaconst}
B(x)= 
\begin{cases}
a(x)  \delta(x)^{\beta -2}  
 \Big[\big(1-q(x)\big)\Big(\frac{\beta+ 
\gamma (x) -1}{2}\Big)^2+  \frac{v(x)}{a(x)}\Big]  & \text{for}\; x\in \Gamma_{\nu_2} \\
0  &\text{otherwise}
\end{cases}.
\end{equation}
Then by choosing 
\begin{equation}
\frac{E_0 -E}{2} = g_{\nu_2},
\end{equation}
\eqref{E:gB} follows from \eqref{E:ESA}, \eqref{E:gnuconst}  \eqref{E:gradgMconst} and \eqref{E:Bbetaconst}.
The last step in the proof of Theorem \ref{Mbetaconst} is to verify \eqref{E:supml} in Lemma
\ref{E:supml} where in the definition of $\phi_l$ we take $r_1 <\nu_2$.

 From \eqref{E:lkderiv}, \eqref{E:phil}  and Assumption A:
\begin{equation}
\nabla \phi_l(x)\cdot \big(\bbD(x)\nabla \phi_l(x)\big)=
\big(k_l'(\delta(x)\big)^2\nabla \delta(x)\cdot \big(\bbD(x)\nabla \delta(x)\big) \leq
\frac{4}{r^2_{l+1}} a(x)\delta (x)^{\beta(x)}
\end{equation}
hence uniformly in $l$:
\begin{equation}\label{E:cphi}
\sup_{x\in B_l }\nabla \phi_l(x)\cdot \big(\bbD(x)\nabla \phi_l(x)\big)\delta(x)^{2-\beta} \leq c_{\phi} <\infty.
\end{equation}

Further, from \eqref{E:gradgMconst} (again uniformly in $l$)

\begin{equation}\label{E:gMconst}
\sup_{x\in B_l }\nabla g(x)\cdot \big(\bbD(x)\nabla g(x)\big)\delta(x)^{2-\beta} \leq c_{g} <\infty.
\end{equation}
On the other hand from \eqref{E:gconst} on $\Gamma_{\mu}$:
\begin{equation}\label{E:expg}
e^{2g(x)} =\delta(x)^{2-\beta} .
\end{equation}
Plugging  \eqref{E:cphi},\eqref{E:gMconst}, \eqref{E:expg} into \eqref{E:|m|} one obtains
(recall that $\supp |m|_l \subset B_l \subset \Gamma_{\mu}$) \eqref{E:suplm}
and the proof of Theorem \ref{Mbetaconst} i. is finished.

ii. Rewrite \eqref{E:hVbigger} as
\begin{align}\label{E:hVqbigger}
&h\big[\varphi,\varphi\big] \geq -\big( c_1(\nu) +
V_{\nu}\big)\| \varphi \|^2+     
\int_{\Gamma_{\nu}}a(x)  \delta(x)^{\beta -2}  
 \Big[\Big(\frac{\beta+ 
 \gamma (x) -1}{2}\Big)^2+ \frac{v(x)}{a(x)}-\nonumber \\  
 & -q(x)\Big(\frac{\beta+ 
 \gamma (x) -1}{2}\Big)^2 + \frac{1-q(x)}{8}\frac{1}{
\ln \Big(\frac{1}{\delta(x)}\Big)^2}
 - \frac{w_{\mu}}{a_{\mu}}\delta(x)^{\frac{1-s}{2}}\Big] \big|\varphi (x)\big|^2 \rho(x)dx.
\end{align}
From \eqref{E:qsmall} (recall that $\gamma \in L^{\infty}(\Gamma_{\nu_0}$) it follows again that
there exists 
$0<\nu_2 < \min\{\mu, \nu_1\}$
such  that on $\Gamma_{\nu_2}$  
\begin{equation}\label{E:lnconstq}
-q(x)\Big(\frac{\beta+ 
 \gamma (x) -1}{2}\Big)^2+\frac{1-q(x)}{8}\frac{1}{
\Big(\ln \frac{1}{\delta(x)}\Big)^2}
 - \frac{w_{\mu}}{a_{\mu}}\delta(x)^{\frac{1-s}{2}} \geq 0.
\end{equation}
The rest of the argument goes unchanged.

\end{proof}

We state now our main result concerning essential self-adjointness oh $H$ in the case when 
neither $\inf_{x\in \Gamma_{\nu_0}}\beta(x)
\geq 2$  nor $\beta(x) = constant$  holds true.  For $\eta < \nu_0$ let
\begin{equation}\label{E:gammaeta}
\Gamma_{\eta,-}=\big\{x\in \Gamma_{\eta} \big| \beta(x) <2\big\}; \quad
\Gamma_{\eta,+}=\big\{x\in \Gamma_{\eta} \big| \beta(x) \geq 2\big\}.
\end{equation}

\begin{theorem}\label{M}
Suppose Assumptions A and Q hold true,
 and in addition there exists
 \begin{equation}\label{E:muM}
 0<\mu <\frac{\nu_0}{2}
 \end{equation}
 
 such that
\begin{equation}\label{E:DM}
\sup_{x\in \Gamma_{\mu,-}}\|\bbD(x)\|\delta(x)^{-\big(\beta(x) +s_{\beta}-s\big)} <\infty,
\end{equation}
\begin{equation}\label{E:VM}
V(x)\geq 
\begin{cases}
0 \quad & \text{for} \;x \in \Gamma_{\mu,+}\\
 V(x)\geq 
 \delta(x)^{\beta-2)}\\\Big[v(x)-w_{\mu}\delta(x)^\frac{1-s}{2}\Big];\; 0\leq w_{\mu} <\infty;\;
  v\in L^{\infty}(\Gamma_{\mu}; \bbR)
  & \text{for}\;x\in  \Gamma_{\mu,-}
\end{cases}.
\end{equation}
Then if
\begin{equation}\label{E:ESAM}
\inf_{\Gamma_{\mu,-}} \frac{\big(1-q(x)\big)\big(\beta(x)+\gamma(x)-1\big)^2  +
\frac{4v(x)}{a(x)}}{\big(\beta(x)-2\big)^2 } \geq 1
\end{equation}
$H$ is essentially self-adjoint.
\end{theorem}

\textit{Proof} The proof of Theorem \ref{M} follows the proof of Theorem \ref{Mbetaconst}.
 As in the proof of Theorem \ref{Mbetaconst} we use Lemma \ref{BI} and Lemma \ref{E:supml} but in this case we need a more elaborate
choice of $g(x)$. 
The chioce of $g(x)$ is as folloows:
\begin{equation}\label{E:gM}
g(x)= \min \big\{ 0, \tilde g(x)\big\}
\end{equation}
with
\begin{equation}\label{E:tildegM}
\tilde g(x)=\frac{2-\beta(x)}{2}\Big( \ln \delta(x)+\lambda \ln \ln \frac{1}{\delta(x)}\Big)\Psi(x)
\end{equation}
where $\Psi$ is  the same as in the proof of Theorem \ref{H} (see \eqref{H:Psi}) and
\begin{equation}\label{E:lambda}
0<\lambda \leq \min \Big\{1,\frac{1}{1+ \sup_{x \in \Gamma_{\nu_0}}|\beta(x)-2|}\Big\}.
\end{equation}

On $\Gamma_{\frac{\nu_0}{2}}$ one has
\begin{equation}\label{E:gnu0-}
g(x) =
\begin{cases}
0 \quad & \text{for} \;x \in \Gamma_{\frac{\nu_0}{2},+}\\
\frac{2-\beta(x)}{2}\Big( \ln \delta(x)+\lambda \ln \ln \frac{1}{\delta(x)}\Big)
&\text{for}\;x\in  \Gamma_{\frac{\nu_0}{2},-}
\end{cases}.
\end{equation}
Indeed (recall that $\nu_0 \leq e^{-1}$ and $\Psi\big|_{\Gamma_{\frac{\nu_0}{2}}}=1$) 
since $ \ln \delta(x)+\lambda \ln \ln \frac{1}{\delta(x)}<0$ it follows that
$g(x) <0$  if and only if $\beta(x)-2 <0$. Moreover since $g$ is Lipschitz
continuous, by Rademacher theorem $\nabla g(x)$ exists a.e. and, whenever it exists, equals 
$ \nabla \tilde g(x)$ if $\beta(x)-2 <0$ and vanishes otherwise. The estimation involving $g(x)$ we need is provided by the next lemma.
\begin{lemma}\label{G}
There exist $0\leq  c_4(\nu_0),\; c_5(\nu_0) <\infty$ such that on $\Gamma_{\frac{\nu_0}{2}}$:

\begin{equation}\label{E:gradgMM}
\nabla g(x)\cdot\big(\bbD(x)\nabla g(x)\big)
\begin{cases}
= 0 \quad & \text{for} \;x \in \Gamma_{\frac{\nu_0}{2},+}\\
\leq c_5(\nu_0) +\big(\frac{\beta (x)-2}{2}\big)^2a(x)\delta(x)^{\beta(x)-2} \\\Big[
\big(1-\frac{\lambda}{\ln\frac{1}{\delta(x)}}\big)^2 +c_4(\nu_0) \delta(x)^{\frac{1-s}{2}}\Big]
 & \text{for}\;x\in  \Gamma_{\frac{\nu_0}{2},-}
\end{cases}.
\end{equation}

\end{lemma}
\begin{proof}
We estimate $\nabla g \cdot\big(\bbD\nabla g\big)$ for $x\in  \Gamma_{\frac{\nu_0}{2},-}$.
From \eqref{E:gnu0-} on  $\Gamma_{\frac{\nu_0}{2},-}$:
\begin{equation}\label{E:gradgM}
\nabla g(x)= \frac{2-\beta(x)}{2}\Big( 1-\frac{\lambda }{ \ln \frac{1}{\delta(x)}}\Big)\delta(x)^{-1} \nabla \delta (x)
-\frac{1}{2}\Big( \ln \delta(x)+\lambda \ln \ln \frac{1}{\delta(x)}\Big) \nabla \beta(x)
\end{equation}
hence from \eqref{S:AD} and Assumption A (recall that $\big|\nabla \delta(x)\big|=1$)
\begin{equation}\label{E:gradgMT}
\nabla g(x) \cdot\big(\bbD(x)\nabla g(x)\big)=\Big(\frac{2-\beta(x)}{2}\Big)^2 a(x)\Big( 1-\frac{\lambda }{ \ln \frac{1}{\delta(x)}}\Big)^2
\delta(x)^{\beta(x)-2} +T(x)
\end{equation}
with
\begin{align}\label{E:T}
&T(x)= \nonumber \\
&-\frac{2-\beta(x)}{4}\Big( 1-\frac{\lambda }{ \ln \frac{1}{\delta(x)}}\Big)\Big( \ln \delta(x)+\lambda \ln \ln \frac{1}{\delta(x)}\Big)\delta(x)^{-1}
\Big[ \nabla \delta (x)\cdot \big(\bbD(x)\nabla \beta(x)\big)+\nonumber \\
&\nabla \beta (x)\cdot \big(\bbD(x)\nabla \delta(x)\big)\Big]
+\frac{1}{4}\Big( \ln \delta(x)+\lambda \ln \ln \frac{1}{\delta(x)}\Big)^2 \nabla\beta (x)\cdot \big(\bbD(x)\nabla \beta(x)\big).
\end{align}
Further, one has to estimate $T$. Consider first the case
\begin{equation}\label{E:betaplus}
\beta (x)-2 \geq \frac{s-1}{2}.
\end{equation}
From Assumption A,  \eqref{S:cA}, \eqref{E:DM} and \eqref{E:T}, counting the powers of $\delta(x)$ one arrives at the conclusion 
that
\begin{equation}
\sup_{ x \in  \Gamma_{\frac{\nu_0}{2}};\;0>\beta (x)-2 \geq \frac{s-1}{2} }\big|T(x)\big| =: c_5(\nu_0)<\infty,
\end{equation}
hence from \eqref{E:gradgMT}
\begin{equation}\label{E:c5}
\nabla g(x) \cdot(\bbD(x)\nabla g(x))\leq \Big(\frac{2-\beta(x)}{2}\Big)^2 a(x)\Big( 1-\frac{\lambda }{ \ln \frac{1}{\delta(x)}}\Big)^2
\delta(x)^{\beta(x)-2}  +c_5(\nu_0).
\end{equation}
Consider now the alternative case
\begin{equation}
\beta (x)-2 <\frac{s-1}{2}
\end{equation}
which implies
\begin{equation}\label{E:betaminus}
\big| \beta(x)-2 \big| >\frac{1-s}{2}
\end{equation}
In this case  from \eqref{E:betaminus}, taking into account that $\frac{1}{a} \in L^{\infty}\big(\Gamma_{\frac{\nu_0}{2}}\big)$,
by counting the powers of $\delta(x)$ as before one obtain that
\begin{equation}
\sup_{ x \in  \Gamma_{\frac{\nu_0}{2}};\;\beta (x)-2 < \frac{s-1}{2} }
\Big(\frac{2}{2-\beta(x)}\Big)^2 \frac{1}{a(x)}\delta(x)^{2-\beta(x)}\delta(x)^{\frac{s-1}{2}} \big|T\big|=:
c_4(\nu_0) <\infty.
\end{equation}
which together with\eqref{E:gradgMT} 

\begin{equation}\label{E:c4}
\nabla g(x) \cdot\big(\bbD(x)\nabla g(x)\big)\leq \Big(\frac{2-\beta(x)}{2}\Big)^2 a(x)
\delta(x)^{(\beta(x)-2)} \Big[\Big( 1-\frac{\lambda }{ \ln \frac{1}{\delta(x)}}\Big)^2+c_4(\nu_0)\Big].
\end{equation}
Finally, \eqref{E:gradgM} follows from \eqref{E:c5} and \eqref{E:c4}.
\end{proof}

We estimate now $h\big[ \varphi, \varphi\big]$ from below.  From  \eqref{E:Vnu}, \eqref{E:VM} and Theorem \ref{H} for all $0<\nu <\min\{\nu_1, \mu\} $

\begin{align}\label{E:hV1}
& h\big[ \varphi, \varphi\big] \geq -\big(c_1(\nu)+V_{\nu} \big) \| \varphi \|^2+\int_{\Gamma_{\nu,-}}a(x)  \delta(x)^{\beta (x)-2} \nonumber \\ 
&\Big[\big(1-q(x)\big) \Big(\frac{\beta+ 
 \gamma (x) -1}{2}\Big)^2+ \frac{v(x)}{a(x)}+\frac{1-q(x)}{8}\frac{1}{\Big( \ln\frac{1}{ \delta(x)}\big)^2 }
 - \frac{w_{\nu}}{a(x)}\delta(x)^{\frac{1-s}{2}}\Big] \big|\varphi (x)\big|^2 \rho(x)dx.
\end{align} 
Since $\frac {1}{\ln \frac{1}{\delta(x)}} \rightarrow 0$ as $\delta(x)  \rightarrow 0$ slower than any
(strictly) positive power of $\delta(x)$ there exists 

\begin{equation}\label{E:nu3}
0<\nu_3 \leq \min\{\nu_1, \mu\}
\end{equation}
such that on $\Gamma_{\nu_3}$  
\begin{equation}\label{E:va}
\frac{1-q(x)}{8}\frac{1}{\Big( \ln\frac{1}{ \delta(x)}\big)^2 }
 - \frac{w_{\nu}}{a(x)}\delta(x)^{\frac{1-s}{2}} \geq 0
\end{equation}
and
\begin{equation}\label{E:ln}
\Big( 1-\frac{\lambda}{\ln \frac{1}{\delta(x)}}\Big)^2
+
 c_4(\nu_0)\delta(x)^{\frac{1-s}{2}}
 \leq 1.
\end{equation}.

From \eqref{E:hV1} and \eqref{E:va}, \eqref{E:h} holds true by choosing:
\begin{equation}\label{E:E0M}
E_0 =-\big(c_1(\nu_3)+V_{\nu_3} \big),
\end{equation}
\begin{equation}\label{E:BM}
B(x)= 
\begin{cases}
a(x)  \delta(x)^{\beta (x)-2} 
\Big[\big(1-q(x)\big)\big(\frac{\beta(x)+ 
 \gamma (x) -1}{2}\big)^2 + \frac{v(x)}{a(x)}\Big]
 &  \text{for}\; x\in \Gamma_{\nu_3,-} \\
0  &\text{otherwise}
\end{cases}
\end{equation}

Further, let
\begin{equation}\label{E:gnu1M}
\sup_{x\in \Omega \setminus \Gamma_{\nu_3}}
\nabla g(x)\cdot \big(\bbD(x) \nabla g(x)\big)=: g_{\nu_3} <\infty.
\end{equation}
Then by choosing
\begin{equation}
\frac{E_0 -E}{2} = c_5(\nu_0) +g_{\nu_3},
\end{equation}
\eqref{E:gB} follows from \eqref{E:ESAM}, \eqref{E:gradgMM}, \eqref{E:ln} and \eqref{E:BM}.

Finally we have to verify Lemma \ref{E:supml}. The choice of $\phi_l$ is as in the proof of Theorem \ref{Mbetaconst} with $r_1 \leq \nu_3$. 
From \eqref{E:lkderiv}, \eqref{E:phil}  and Assumption A (recall also that by construction
on $B_l$, $\delta(x) >r_l$ and $a \in L^{\infty}\big(\Gamma_{\nu_0}\big)$):
\begin{align}\label{E:nablaphi}
&\nabla \phi_l(x)\cdot \big(\bbD(x)\nabla \phi_l(x)\big)= 
\big(k_l'(\delta(x)\big)^2\nabla \delta(x)\cdot \big(\bbD(x)\nabla \delta(x)\big) \leq
\nonumber \\
&\frac{4}{r^2_{l+1}} a(x)\delta (x)^{\beta(x)} \leq 16a(x)\delta^{\beta(x)-2} \leq C_{\phi}
\delta^{\beta(x)-2}.
\end{align}

Let $B_{l,+}=\big\{x\in B_l \big| \beta(x) \geq 2\big\},\;
B_{l,-}=\big\{x\in B_l \big| \beta(x) < 2\big\}$.

From Lemma \ref{G} and \eqref{E:nablaphi}

\begin{equation}\label {E:ml+}
\big| m_l \big| \Big|_{B_{l,+}}= \nabla \phi_l\cdot \big(\bbD\nabla \phi_l\big)\Big|_{B_{l,+}} \leq C_{\phi}.
\end{equation}
Further, for $x \in B_{l,-}$,   from Lemma \ref{G} there exists $c_6(\nu_0) < \infty$
such that 
\begin{equation}\label{E:gradgM-}
\nabla g(x) \cdot\big(\bbD(x)\nabla g(x)\big) \leq c_5(\nu_0) +c_6(\nu_0)
\delta(x)^{\beta(x)-2},
\end{equation}
and from \eqref{E:gM}, \eqref{E:lambda}
\begin{equation}\label{E:expg-}
e^{2g(x)} \leq e^{2-\beta(x)}\ln \frac{1}{\delta(x)} .
\end{equation}
Plugging \eqref{E:nablaphi}, \eqref{E:ml+},  \eqref{E:gradgM-}, \eqref{E:expg-} into \eqref{E:|m|} and taking into account that on $B_l$, 
$\delta(x) \geq \nu_1 2^{1-l}$ (see \eqref{E:rl} and \eqref{E:Bl}) one obtains \eqref{E:suplm}
and the proof of Theorem \ref{M} is finished.

\section{ Comments, extensions, open problems} \label{S:5}
In this section besides comments and applications  we give some extensions of the results in previous sections. We would like to stress that the proofs of these extensions follow closely the proofs in Section \ref{S:3} and Section \ref{S:4}. What has to be added in each case is the choice of the vector field,
$\bsX(x)$ in Lemma \ref{H:VFA} and/or the choice of the function $g(x)$ in the Agmon type exponential estimation as given Lemma \ref{BI} leading via Lemma \ref{E:supml} to essential self-adjointness. The choices for $g(x)$ in this paper are extensions to the case at hand of those used in \cite{NN1} for the case $ \bbD =\mathds{1},\; \rho =1$.

1. For the sake of discussion let us fix (remind) the terminology. Let $\Omega \subset \bbR^d,\; H_0$ as given in Section \ref{S:2} and $\cH:\Omega \rightarrow [0,\infty),\; \cH \in L^{\infty}_{loc}(\Omega)$. $H_0$ is said to satisfy a $\cH$-weak Hardy inequality if there exists $M>-\infty$ such that for all $\varphi \in \C_0^{\infty}(\Omega)$ 
\begin{equation}\label{C:weakHardy}
h\big[ \varphi, \varphi\big] \geq M \int_{\Omega}\big| \varphi(x)\big|^2 \rho(x)dx
+\int_{\Omega} \cH(x)\big|\varphi(x)\big|^2 \rho(x)dx,
\end{equation}
and if in addition $M\geq 0$, $H_0$ is said to satisfy a $\cH$-strong Hardy inequality.
$H_0$ is said to satisfy a $\cH$-boundary layer Hardy inequality if there exists $\nu >0$
(see Section \ref{S:2} for $\Gamma_{\nu}$) such that for all $\varphi \in \C_0^{\infty}(\Gamma_{\nu})$ 
(see \cite{Ro3} for relations between week and boundary layer Hardy inequalities)
\begin{equation}\label{C:blHardy}
h\big[ \varphi, \varphi\big] \geq 
\int_{\Gamma_{\nu}} \cH(x)\big|\varphi(x)\big|^2 \rho(x)dx.
\end{equation}
Notice that the $\cH$-weak Hardy inequality for $H_0$ is equivalent with the 
 $\cH$-weak Hardy inequality for $H_0-M$. Since essential self-adjointness is stable against addition of constants when using Hardy inequalities to ensure essential self-adjointness of $H$, weak Hardy inequalities (which are easier to come by) are sufficient. The function $\cH$ is sometimes called "Hardy barrier" and in order to provide nontrivial lower bound for $h_0\big[ \varphi, \varphi\big]$ it must blow up as $x \rightarrow \partial \Omega$
 
With this terminology, Theorem \ref{H} provides a general $\cH_0$-weak/ boundary layer Hardy inequality for $H_0$. The formula \eqref{H:calH1} gives the height or the Hardy barrier. We would like to stress that $\cH_0(x)$ depends locally on $\bbD$ and $\rho$ i.e. the Hardy barrier is a local quantity. 
 
 In the particular case when $\rho=1$ and in addition $\bbD$ is isotropic (hence $q=0$) i.e.  
 $\bbD(x)=d(x) \mathds{1}=a(x)\delta(x)^{\beta(x)} \mathds{1}$, Theorem \ref{H} gives 
 (for simply connected, bounded domains  with smooth boundary and sufficiently small $\nu$) a  
weak/boundary layer weighted Hardy inequality generalizing previous results in \cite{Ro3}
\begin{align}\label{C:weakH}
& \int_{\Omega}a(x)\delta(x)^{\beta(x)} \big|\nabla \varphi(x) \big|^2 dx
 \geq  \nonumber \\
 &-c_1(\nu)\int_{\Omega}\big|\chi_{\Omega \setminus \Gamma_{\nu}}(x)  \varphi(x) \big|^2 dx 
 + \frac{1}{4}\int_{\Gamma_{\nu}}a(x) \big( \beta(x)-1\big)^2 \delta(x)^{\beta(x)-2}\big|  \varphi(x) \big|^2 dx  
 \end{align}
with $c_1(\nu) < \infty$. 
What \eqref{C:weakH} adds to the previous results  is that the degree of degeneracy of $\bbD$
can vary in $\Gamma_{\nu}$ along  $\partial \Omega$.
In the "classical" case  $\bbD =\mathds{1}; \rho=1$ one recovers the famous constant $\frac{1}{4}$ \cite{BM}:
\begin{align}\label{C:calH011}
\int_{\Omega}\big|\nabla \varphi(x)\big|^2dx \geq &
 -c_1(\nu_1)\|\chi_{\Omega \setminus\Gamma_{\nu_1}} \varphi \|^2
+\frac{1}{4}
\int_{\Gamma_{\nu_1}} \frac{1}{\delta(x)^2}\big|\varphi(x)\big|^2 dx \geq \nonumber \\ 
& -C\|\varphi \|^2
+\frac{1}{4}
\int_{\Omega} \frac{1}{\delta(x)^2}\big|\varphi(x)\big|^2 dx, \; C< \infty
\end{align}
which is known to be optimal in the sense that there exist simply connected, bounded domains
with $\cC^2$ boundary for which $C$ cannot be negative \cite{MMP}.

There is a large body of literature on weak weighted Hardy inequalities on domains in $\bbR^d$ with smooth/rough boundaries which can be partly traced from references in \cite{Ro3}. 
However the optimality problem for the case of rough boundaries is far from being solved. In particular an interesting open problem is whether  Theorem \ref{H} remains true if $\partial \Omega$ is $\cC^{1+\epsilon}$, $0<\epsilon <1$.

2. As a particular case of a more general result, for $H_0= -\Delta$, $\rho =1$ , $\Omega$ bounded and convex G. Barbatis, 
S. Filippas, and A. Tertikas (see \cite{BFT2} and \cite{NT} for more references and related developments) refined the strong Hardy inequality by adding an infinite series of logarithmic sub-leading terms to the Hardy barrier. By using the approach in \cite{BFT2} one can refine Theorem \ref{H} and give a hierarchy of weak anisotropic Hardy inequalities by adding logarithmic subleading terms to $\cH_0(x)$.

The technical tool is a class of functions related to iterated logarithms.
Given   $a \in (0,\infty)$ and $p=1,2,...$ let $t_{a,p}$ and   $L_p: (0, t_{a,p}] \rightarrow (0, 1]$ given recurrently by
\begin{equation}\label{C:tap}
t_{a,1}=\frac{a}{e},\; t_{a,p}=\frac{a}{e^{\frac{1}{t_{a,p-1}}}},
\end{equation}
\begin{equation}\label{C:Lap}
L_{a,1}(t)= \frac{1}{\ln\frac{a}{t}};\; L_{a,p}(t)= L_{a,1}\big(L_{a,p-1}(t)\big)
\end{equation}
and then
\begin{equation}\label{C:Mf}
M_{a,p}(t)=\prod_{k=1}^p L_{a,k}(t); \;f_{a,p}(t)= \sum_{l=1}^p M_{a,l}(t).
\end{equation}
By direct computation the solution of recurrence \eqref{C:Lap} is
\begin{equation}\label{C:iterln}
L_{a,p}(t)= \frac{1}{\ln a \ln a...\ln a\frac{1}{t}}
\end{equation}
and from \eqref{C:iterln} it is not hard to see that $L_{a,p}$ are well defined on $(0,t_{a,p}]$
and take values in $(0,1]$.
We also notice that $t_{1,p}=\frac{1}{e^{e^{\cdot^{\cdot^{e}}}}}$, $t_{e,p}=1$ and
\begin{equation}
t_{a,p}
\begin{cases}
<1;\;\lim_{p \rightarrow \infty}t_{a,p}=0 & \text{for} \; a<e\\
>1;\; \lim_{p \rightarrow \infty}t_{a,p}=\infty & \text{for} \; a>e
\end{cases}
\end{equation}

The magic of these functions lies in the following lemma ( see formula (2.4) in \cite{BFT2} where the case $a=e$ was considered).
\begin{lemma}\label{C:BFT}
For N=1,2,...
\begin{equation}\label{C:fM}
tf'_{a,N}(t)-\frac{1}{2}f_{a,N}(t)^2=\frac{1}{2} \sum_{k=1}^NM_{a,k}(t)^2.
\end{equation}
\end{lemma}
\begin{proof} See \cite{BFT2}. However for the sake of the reader we give it below. The proof is by induction. For $N=1$, $f_{a,1}=L_{a,1}$ and  \eqref{C:fM} follows by direct computation.
Suppose \eqref{C:fM} holds true for $f_{a,N}$. By direct computation:
\begin{equation}\label{C:MN+1der}
L_{a,l}^{'}= \frac{1}{t}L_{a,l}M_{a,l}; \;M_{a,N+1}^{'}(t)= \frac{1}{t}M_{a,N+1}(t)f_{a,N+1}(t).
\end{equation}
Further from \eqref{C:Mf}
\begin{equation}\label{C:fN+1fN}
f_{a,N+1} =f_{a,N} +M_{a,N+1}
\end{equation}
hence from \eqref{C:MN+1der}
\begin{equation}
f'_{a,N+1}(t) = f'_{a,N}(t)+\frac{1}{t}M_{a,N+1}(t)f_{a,N+1}(t)
\end{equation}
which together with \eqref{C:fN+1fN} and \eqref{C:fM} leads to
\begin{equation}
tf'_{a,N+1}(t)-\frac{1}{2}f_{a,N+1}(t)^2=\frac{1}{2} \sum_{k=1}^{N+1}M_{a,k}(t)^2.
\end{equation}
\end{proof}

For definiteness we formulate the refinement of Theorem \ref{H} in the case $a=e$ and omit the subscript $e$ i.e.  we write 
$M_p(t),\; f_p(t)$ for $M_{e,p}(t)$ and  $f_{e,p}(t)$
respectively.

\begin{theorem}\label{Href}
Suppose that Assumptions A and  Q hold true. 
Then  for $N=1,2, ...$   there exists $0<\nu_N <\frac{\nu_0}{2}$ such that  for all $0<\nu < \nu_N$ 
and  $\varphi
 \in \cC^{\infty}_0(\Omega)$:
 
\begin{equation}\label{H:HardyN} 
h_0\big[\varphi, \varphi\big]
\geq -c_N(\nu) \|\chi_{\Omega \setminus \Gamma{\nu}}\varphi \|^2+ \int_{\Gamma_{\nu}}\cH_N(x)|\varphi(x)|^2 \rho(x) dx
\end{equation}
where $c_N(\nu) <\infty$,
$\chi_{\Omega \setminus \Gamma{\nu}}$ is the characteristic function of $\Omega \setminus \Gamma{\nu}$ and
\begin{equation}\label{H:calHN}
\cH_N(x)= \frac{1}{4}\big(1-q(x)\big)a(x)\delta(x)^{\beta(x)-2}\Big[  \big(\beta(x)+\gamma(x)-1\big)^2 +
\sum_{k=1}^NM_k\big(\delta(x)\big)^2 \Big].
\end{equation}
In particular if $\varphi \in \cC^{\infty}_0(\Gamma_{\nu})$
\begin{equation}\label{H:HardyLN}
h_0[\varphi,\varphi] \geq \int_{\Gamma_{\nu}}\cH_N(x)|\varphi(x)|^2 \rho(x) dx.
\end{equation} 
\end{theorem}

\begin{proof}
Repeat the proof of Theorem \ref{H}. Replace $\widetilde\bsX_0$ given by \eqref{H:tildeX1}
with $\widetilde\bsX_N$  given by the same formula with $f\big(\delta(x)\big)$ replaced by $f_{N+1}\big(\delta(x)\big)$.
The same computation as in the proof of Theorem \ref{H} together with Lemma \ref{C:BFT} leads to:
\begin{align}\label{H:divXND-1XNRN}
&\nabla \cdot \bsX_N(x) -\bsX_N(x)\cdot \big(\rho (x) \ID(x)\big)^{-1}\bsX_N(x))=
\big(1-q(x)\big)a(x)\rho(x)\delta(x)^{\beta(x)-2}\nonumber \\
&\Big[\frac{1}{4}\big(\beta(x)+\gamma(x)-1\big)^2
+\frac{1}{4} \sum_{k=1}^NM_k\big(\delta(x)\big)^2 +\frac{1}{4}M_{N+1}\big(\delta(x)\big)^2
+R_{N+1}(x)\Big].
\end{align}
with $R_N(x)$ given by \eqref{H:R1} where $f\big(\delta(x)\big)$ is replaced by $f_{N+1}\big(\delta(x)\big)$. $R_{N+1}(x)$ satisfies \eqref{H:supR1} hence there exists $0 <\nu_{N} < \frac{\nu_0}{2}$ such that on $\Gamma_{\nu_N}$
\begin{equation}
\frac{1}{4}M_{N+1}\big(\delta(x)\big)^2
+R_{N+1}(x) \geq 0.
\end{equation}
\end{proof}

The subleading term, $\sum_{k=1}^NM_k\big(\delta(x)\big)^2$, in 
\eqref{H:calHN} becomes relevant when  the leading term  vanishes. An example where this is the case is the "standard" weak weighted/ boundary level
Hardy inequality i.e. $\bbD(x) =\delta(x)^{\beta}\mathds{1}, \;\rho(x)=1$ when \eqref{H:calH1} becomes
(see also \cite{Ro3})
\begin{equation}\label{C:weakweighH}
h_0\big[\varphi, \varphi\big]
\geq -c_N(\nu) \|\chi_{\Omega \setminus \Gamma{\nu}}\varphi \|^2+ \Big(\frac{\beta-1}{2}\Big)^2\int_{\Gamma_{\nu}}\delta(x)^{\beta-2}|\varphi(x)|^2 \rho(x) dx.
\end{equation}
Notice that for $\beta =1$ \eqref{C:weakweighH} gives no information while Theorem \ref{Href}
gives (for sufficient small $\nu$)

\begin{equation}\label{H:HardyNref}
h_0\big[\varphi, \varphi\big]
\geq -c_N(\nu) \|\chi_{\Omega \setminus \Gamma{\nu}}\varphi \|^2+ \frac{1}{4}
\int_{\Gamma_{\nu}} \frac{1}{\delta(x)}\big( \sum_{k=1}^NM_k\big(\delta(x)\big)^2\big) |\varphi(x)|^2 \rho(x) dx.
\end{equation}

3.  Theorem \ref{H} can be generalised to the case when logarithmic factors are 
added in the behaviour of $d(x)$ (see \eqref{S:blockD}) as $x \rightarrow \partial \Omega$ and we give below such an extension.

\begin{theorem}\label{C:Hln}
Suppose that Assumptions A  and Q with \eqref{S:AD} replaced with
\begin{equation}\label{S:ADln}
d(x)=a(x)\delta(x)^{\beta(x)}\big(\ln \delta(x)\big)^{\alpha}; \quad  a \in \cC^1\big(\Gamma_{\nu_0}; (0,\infty)\big);  \quad \beta \in 
\cC^1\big(\Gamma_{\nu_0}; \bbR\big), \; \alpha \in \bbR.
\end{equation}
hold true.
Then there exists $0 <\nu_{\alpha} <\frac{\nu_0}{2}$ such that for
for all $0<\nu < \nu_{\alpha}$ 
and  $\varphi
 \in \cC^{\infty}_0(\Omega)$:
 
\begin{equation}\label{H:Hardyalpha} 
h_0\big[\varphi, \varphi\big]
\geq -c_1(\nu) \|\chi_{\Omega \setminus \Gamma{\nu}}\varphi \|^2+ \int_{\Gamma_{\nu}}\cH_{\alpha}(x)|\varphi(x)|^2 \rho(x) dx
\end{equation}
where $c_1(\nu) <\infty$,
$\chi_{\Omega \setminus \Gamma{\nu}}$ is the characteristic function of $\Omega \setminus \Gamma{\nu}$ and
\begin{equation}\label{H:calHalpha}
\cH_{\alpha}(x)= \frac{1}{4}\big(1-q(x)\big)a(x)\delta(x)^{\beta(x)-2} \Big[ \big(\beta(x)+\gamma(x)-1 -\frac{\alpha}{\ln\frac{1}{\delta(x)}}\big)^2
+\frac{1-2\alpha -\alpha^2}{\ln\frac{1}{\delta(x)}^2}\Big].
\end{equation}
In particular if $\varphi \in \cC^{\infty}_0(\Gamma_{\nu})$
\begin{equation}\label{H:HardyLalpha}
h_0[\varphi,\varphi] \geq \int_{\Gamma_{\nu}}\cH_{\alpha}(x)|\varphi(x)|^2 \rho(x) dx.
\end{equation} 
\end{theorem}

\begin{proof}
Let $\tilde\nu_0 = \min \{\nu_0, \frac{1}{e^e}\}$ and on $\Gamma_{\tilde\nu_0}$ 
\begin{equation}\label{H:tildeXalpha}
\widetilde\bsX_{\alpha}= \frac{1}{2}\big(1-q(x)\big)a(x)r(x)\delta(x)^{\beta(x)+\gamma(x)-1}
\big(\ln\frac{1}{\delta(x)}\big)^{\alpha}
\big[\beta(x)+\gamma(x)-1 + f_{1,2}\big(\delta(x)\big)\big]\nabla\delta(x)
\end{equation} 
where (see \eqref{C:Lap}, \eqref{C:Mf})

\begin{equation}
f_{1,2}(t)= \frac{1}{\ln\frac{1}{t}} +
\frac{1}{\big(\ln\frac{1}{t}\big)\big(\ln \ln \frac{1}{t}\big)}.
\end{equation}

Defining $\bsX_{\alpha}$ in terms of $\widetilde\bsX_{\alpha}$ as in the proof of Theorem \ref{H} with  $\nu_0$ replaced by $\tilde\nu_0$, repeating  the computations  and using the Lemma \ref{C:BFT}
one obtains that on $\Gamma_{\frac{\tilde\nu_0}{2}}$
\begin{align}\label{H:divXalphaD-1XNRalpha}
&\nabla \cdot \bsX_{\alpha}(x) -\bsX_{\alpha}(x)\cdot \big(\rho (x) \ID(x)\big)^{-1}\bsX_{\alpha}(x))=
\big(1-q(x)\big)a(x)\rho(x)\delta(x)^{\beta(x)-2}
\Big(\ln\frac{1}{\delta(x)}\Big)^{\alpha}\nonumber \\
&\Big\{\Big[\frac{1}{4}\Big(\beta(x)+\gamma(x)-1-\alpha \frac{1}{\ln\frac{1}{\delta(x)}}\Big)^2+
\frac{1-2\alpha -\alpha^2}{\ln\frac{1}{\delta(x)}^2}+ \nonumber \\
&\Big(\frac{1}{\ln\frac{1}{\delta(x)}}\Big)^2 \frac{1}{\ln \ln \frac{1}{\delta(x)}}\Big(1-\alpha \frac
{1}{\ln \ln \frac{1}{\delta(x)}}\big) \Big] +R_{\alpha}(x)\Big\}.
\end{align}
with
\begin{equation}
\sup_{x \in \Gamma_{\frac{\tilde\nu_0}{2}}}\big|R_{\alpha}(x)\big| \delta(x)^{1-s} < \infty.
\end{equation}
Then
 \eqref{H:Hardyalpha} with \eqref{H:calHalpha} follows choosing ${\nu_{\alpha}}$ small enough as  to have on $\Gamma_{\frac{\tilde\nu_{0}}{2}}$
\begin{equation}
\frac{1}{4}\Big(\frac{1}{\ln\frac{1}{\delta(x)}}\Big)^2 \frac{1}{\ln \ln \frac{1}{\delta(x)}}\Big(1-\alpha \frac
{1}{\ln \ln \frac{1}{\delta(x)}} \Big) +R_{\alpha}(x) \geq 0.
\end{equation}
\end{proof}

4.  For the sake of simplicity we assumed that $\Omega$ is bounded, simply connected and with smooth boundary. It is not hard to generalise  Theorems \ref{H} and \ref{Href} to the case when boundedness and simply connectedness conditions are relaxed. Still in order to control the behaviour at infinity uniformity conditions are to be imposed. We follow \cite{Br} at this point. More precisely, let $\Gamma \subset
\bbR^d$ be a $k$-dimensional, \\$k=0,1,2, ...,d-1$ differentiable manifolds without boundary. To each point $x^* \in \Gamma$ we associate a Cartesian system of coordinates in which the first $k$ coordinate axes are in the tangent space to $\Gamma$ at $x^*$. For given $0< r,\; C< \infty$ we say that $\Gamma \in \cF_{r,C}$ if for every $x^* \in \Gamma$ there exists a sphere a radius $r$ with centre at $x^*$ such that in this system of coordinates $\Gamma \cap \big\{ \big| x-x^*\big| <r\big\}$ 
is given by the equations $x_j = f_j\big(x_1,...,x_k \big);\;j=k+1,...,d$ where $ f_j \in \cC\big( \bbR^d; \bbR\big)$ and $\big| D^{\alpha}f_j(x_1,...,x_k)\big| \leq C$ for all multiindices 
$\alpha, \; |\alpha | \leq 2$.

For example, in this setting Theorem \ref{H} takes the form:
\begin{theorem}\label{HMC}
Let $\Omega \subset \bbR^d$ such that  $\partial \Omega =\cup_{j=1}^{N} \Gamma_j,\; N \leq \infty; \;
\Gamma_j \in  \cF_{r,C}; \; \text{dist} \big(\Gamma_j,\Gamma_k\big) \geq d_0 >0$.

Suppose that Assumption A and Q hold true with $\nu_{\Omega} \leq \frac{d_0}{4}$ , and in addition

\begin{equation}
(1-q)^{-1}, \; a,\; r,\; \frac{1}{a},\; \frac{1}{r},\; \beta, \;\gamma  \in L^{\infty}\big(\Gamma_{\nu_0}\big).
\end{equation}

Then   there exists $0<\nu_1 <\frac{\nu_0}{2}$ such that  for all $0<\nu < \nu_1$ 
and  $\varphi
 \in \cC^{\infty}_0(\Omega)$:

\begin{align}
h_0\big[\varphi, \varphi\big]
=\int_{\Omega}\overline{\nabla\varphi(x)}\cdot \big(\bbD(x)\nabla\varphi(x)\big)\rho(x)dx\nonumber \\
\geq -c_1(\nu) \|\chi_{\Omega \setminus \Gamma_{\nu}} \varphi \|^2+ \sum_{j=1}^{N}\int_{\Gamma_{j,\nu}}\cH_{0,j}(x)|\varphi(x)|^2 \rho(x) dx
\end{align}
where
\begin{equation}
\Gamma_{j,\nu}= \big\{x \in \Omega\big| \dist \big(x,\Gamma_j\big) <\nu \big\}
\end{equation}
and

\begin{equation}\label{H:calH0j}
\cH_{0,j}(x)= \frac{1}{4}\big(1-q(x)\big)a(x)\delta(x)^{\beta(x)-2} \Big[ \big(\beta(x)+\gamma(x)+d-k_j -2\big)^2 +\frac{1}{2}\frac{1}{\Big(\ln \frac{1}{\delta(x)}\Big)^2}\Big].
\end{equation}
In addition if $\varphi \in \cC^{\infty}_0(\Gamma_{j,\nu})$
\begin{equation}
h_0[\varphi,\varphi] \geq \int_{\Gamma_{j,\nu}}\cH_{0,j}(x)|\varphi(x)|^2 \rho(x) dx.
\end{equation}
\end{theorem}

\textit{ Outline of proof.} With the following changes the proof of Theorem \ref{HMC} mimics closely the proof of Theorem \ref{H}.

i. Instead of Lemma \ref{GT} one has to use its generalisation (see Lemma 6.1 and Lemma 6.2
in \cite{Br}, see also Theorem 3.2 in \cite{AS})

\begin{lemma}\label{Br}

There exists $0< \nu_{\Omega} \leq \frac{d_0}{4}$ and $C_{\Delta} <\infty$ such that
$\delta \in \cC^2\big(\Gamma_{\nu_{\Omega}}\big)$ and for 
 $x \in \Gamma_{j,\nu_{\Omega},},$

 \begin{equation}
\big| \nabla \delta (x)\big|  =1; \quad \big| \Delta \delta (x) - \frac{d-k_j -1}{\delta (x)}  \big| \leq C_{\Delta}.
\end{equation}
where $k_j$ is the dimension of $\Gamma_j$.
\end{lemma}

ii.  The choice of  $\widetilde\bsX_0(x)$ in \eqref{H:tildeX1} is now as follows.

For  $x \in \Gamma_{j,\nu_{0}},$:

\begin{equation}
\widetilde\bsX_0= \frac{1}{2}\big(1-q(x)\big)a(x)\rho(x)\delta(x)^{\beta(x)-1}\big(\beta(x)+\gamma(x)+d -k_j-2 +f\big(\delta(x)\big)\big)\nabla\delta(x).
\end{equation}

Accordingly, for $x \in \Gamma_{j,\nu_{0}} \; R_0(x)$ in \eqref{H:divX1} rewrites as:
\begin{equation}
R_0(x)=\frac{1}{2}\big(\beta(x)+\gamma(x)+d -k_j-1\big) + \mathcal O\big(\delta (x)^\frac{1-s}{2}\big). 
\end{equation}

5. Together with \eqref{H:calH1} giving the Hardy barrier, the inequalities \eqref{E:ESA},  \eqref{E:ESAM} which, for domains with   $\cC^2$ boundary, shows how the coefficients of $H$ merge together to ensure  
essential self-adjointness, are the main results of this paper. As in the case of Hardy barrier the essential self-adjointness condition is local in the sense that it is expressed in terms of the infimum over boundary layers of local quantities.
Notice that one can have essential self-adjointness of $H$ even for negative (deconfining!) potentials and essential self-adjointness of $H_ 0$
when $\beta(x)$  is negative i.e. $\bbD(x)$ blows up as $x \rightarrow \partial \Omega$. What our results add to the previous ones is on one hand that the degree, $\beta(x)$, of degeneracy of $\bbD(x)$ is allowed to vary over $\partial \Omega$ thus giving  a positive answer to a question in \cite{Ro1}
and on the other hand the isotropy of $\bbD(x)$ in the limit $x \rightarrow \partial \Omega$ is weakened to the assumption that it has near $\partial \Omega$ a $x$-dependent $2 \times 2$
block structure dominated by the diagonal (see \eqref{S:blockD}, \eqref{H:q} and \eqref{S:Q}).

6. Having the Hardy barrier from Theorem \ref{HMC} one can extend also the criteria for essential self-adjointness to unbounded multiply connected domains with $\cC^2$ boundary. For simplicity we only give 
the analog of Theorem \ref{Mbetaconst} for $V = 0$.

\begin{theorem}
In the setting of Theorem \ref{HMC} suppose suppose that
\begin{equation}\label{C:Dinfty}
\big(1+ | \cdot |\big)^2 \bbD \in  L^{\infty}\big( \Omega \setminus \Gamma_{\nu_0}\big),
\end{equation}

  and in addition there exists
 \begin{equation}\label{E:mu}
  0<\mu <\frac{\nu_0}{2}
 \end{equation}
 such that
  on $\Gamma_{\mu,j}$ either:
  
i. 
\begin{equation}\beta(x) = \beta_j \geq 2
\end{equation}
or

ii

\begin{equation}
\beta(x)  = \beta_j < 2, \quad \inf_{x \in \Gamma_{\mu,j}} \big(1-q(x)\big)
\Big (\frac{\beta_j+\gamma(x) + d   -k_j -2}{\beta_j -2}\Big)^2  \geq 1.
\end{equation}
Then $H_0$ is essentially self-adjoint.
\end{theorem}
\begin{proof}
The verification of Lemma \ref{BI}  and Lemma \ref{E:supml} follows the one in Theorem
  \ref{ESAC} in the case i. and the one in Theorem \ref{Mbetaconst} in the case ii. respectively. Since
for unbounded domains $\phi_l$ as defined in \eqref{E:phil} do not have compact support
one has to replace \eqref{E:phil} with
\begin{equation}\label{E:philunb}
\phi_l(x) = k_l\big(\delta(x)\big)k_{\infty}\big(\frac{|x|}{l}\big)
\end{equation}
where $k_l$ is given by \eqref{E:kl} and
\begin{equation}\label{E:kinfty}
k_\infty(t)\in \cC^1\big((0,\infty)\big);\quad \supp k'_{\infty} \subset (1,2) ;\quad
k_{\infty}(t)=
\begin{cases}
0\,, &\quad t>2\\
1\,, &\quad t<1,
\end{cases}
\end{equation}
\begin{equation}\label{E:kinftyderiv}
\big|k'_{\infty}(t)\big| \leq 2.
\end{equation}

Now, notice that as $|x| \rightarrow \infty$, $ \big|\nabla k_{\infty}\big(  \frac{|x|}{l}\big)\big| \sim
\frac{1}{l}$ which together with \eqref{C:Dinfty} implies the boundedness of
$\nabla \phi_l \cdot \big(\bbD\nabla \phi_l\big)$ as $|x| \rightarrow \infty$.

\end{proof}

7. Building on methods and results in \cite{NN2} and \cite{Ro1} D. W.  Robinson \cite{Ro2}
made a thorough study of essential self-adjointness of $H_0$ for the case $\rho =1$,
$\beta =const$, $\bbD \in L^{\infty}(\Omega: \bbR^{d \times d})$ and is isotropic in the limit
 $\delta(x) \rightarrow 0$:

\begin{equation}\label{C:Disot}
\inf_{r\in (0,r_0]} \sup_{\Gamma_r} \|\bbD(x)\delta(x)^{-\beta}- d(x)\mathds{1}\| =0.
\end{equation}

For the sake of simplicity  of  the discussion below we state his main result for the particular case when $\Omega$ is bounded, simply connected and with $\cC^2$ boundary.

\begin{theorem}\label{Ro} 
Suppose that $\bbD$ is Lipschitz continuous and there exists $r_0 >0$ such 
that \eqref{C:Disot} holds true with $d$ Lipschitz continuous $d,\; \frac{1}{d},\;
|\nabla d| \in L^{\infty}\big(\Gamma_{r_0}\big)$. Then:

i. For $\beta >\frac{3}{2}$  $H_0$ is essentially self-adjoint.

ii. If in addition

\begin{equation} 
\sup_{\Gamma_{r_0}}\big| \big( \nabla\bbD\delta\big)(x)\cdot\nabla \delta(x)\big| <\infty,
\end{equation}
where 
\begin{equation}
\big( \nabla\bbD\delta^{-\beta}\big)_j:= \sum_{k=1}^{d}\partial_k \big( \nabla\bbD\delta^{-\beta}\big)_{kj}, 
\end{equation}
 $H_0$ is not essentially self-adjoint for $\beta <\frac{3}{2}$.
\end{theorem}

Theorem \ref{Ro} provides the answer (at the level of  power behaviour) to the question
 of dependence of essential self-adjointness as a function of 
the degree of degeneracy of $\bbD(x)$ as $x \rightarrow \partial\Omega$, except for the 
critical value, $\beta =\frac{3}{2}$, which remains undecided. At the price of  a quantitative version of \eqref{C:Disot} the following proposition  fills this gap. 
 
 \begin{proposition}\label{PR}
Let  $\Omega$ is bounded, simply connected and with $\cC^2$ boundary. Suppose
that $\bbD \in \cC^1(\Omega; \bbR^{d \times d})$ and there exists $r_0 >0,\; s<1$ such that on 
$\Gamma_{r_0}$
 
 \begin{equation}\label{C:Disots}
  \|\bbD(x)\delta(x)^{-\beta}- d(x)\mathds{1}\| \leq c_{r_0} \delta(x)^{1-s} 
\end{equation}
with $d \in \cC^1(\Omega)$ and $d,\; \frac{1}{d},\;
|\nabla d| \in L^{\infty}\big(\Gamma_{r_0}\big)$.
Then $H_0$ is essentially self-adjoint for $\beta \geq \frac{3}{2}$.
\end{proposition}
\begin{proof}
Since the case $\beta \geq 2$ is covered by Theorem \ref{ESAC} one can assume that $\beta <2$.
Remark first that from \eqref{C:Disots} and 
 $d,\; \frac{1}{d},\;
 \in L^{\infty}\big(\Gamma_{r_0}\big)$
it follows that if $\widetilde\bbD= d \delta^{\beta}$ then
\begin{equation}\label{C:Dper}
\bbD= \widetilde\bbD + \delta^{\beta} \cO\big(\delta^{1-s}\big)
\end{equation}
and for sufficiently small $\delta(x)$,
\begin{equation}\label{C:Dinvper}
\bbD^{-1}= \widetilde\bbD^{-1} + \delta^{-\beta} \cO\big(\delta^{1-s}\big).
\end{equation}
Repeat the proof of Theorem \ref{H} with $\widetilde\bsX_0$ corresponding to $\widetilde\bbD$. 
i.e. $\widetilde\bsX_0 =\frac{1}{2}\big(1-q(x)\big)d(x)\delta(x)^{\beta-1}\big( \beta -1 +
f\big(\delta(x)\big)\big) \nabla\delta(x)$. 
Notice that the "perturbation" $\cO\big(\delta^{1-s}\big)$ enter only in the computation 
$\widetilde\bsX_0 \cdot \big(\bbD^{-1}
\widetilde\bsX_0\big) =  \widetilde\bsX_0 \cdot \big(\widetilde\bbD^{-1} \widetilde\bsX\big)
+\delta^{\beta -2}\cO\big(\delta^{1-s}\big)$ hence one obtains \eqref{H:calH1} with
$q=0$. As concerning $g$ due to the perturbation term in \eqref{C:Dper} we need we need a better choice than in Theorem \ref{Mbetaconst}. More precisely we take
\begin{equation}\label{C:gconst}
 g(x)=\Big(\frac{2-\beta}{2} \ln \delta(x)+\frac{1}{2} \ln \ln \frac{1}{\delta(x)}\Big)\Psi(x)
\end{equation}
where $\Psi$ (see \eqref{H:Psi}) is  the same as in the proof of Theorem \ref{H}.
Next, we compute  $\nabla g(x)\cdot \big(\bbD(x)\nabla g(x)\big)$on $\Gamma_{\mu}$.
 For $\delta(x) <\frac{\nu_0}{2}$, $\Psi(x)=1$ hence for $x \in \Gamma_{\mu}$
 \begin{equation} \label{E:gradgconst}
\nabla g(x)= \frac{2-\beta}{2}\delta(x)^{-1} \Big(1-\frac{1}{2 \ln \frac{1}{\delta(x)}}\Big)\nabla\delta(x)
\end{equation}
which together with \eqref{C:Dper} leads to

\begin{equation}\label{E:gradgMconstiso}
\nabla g(x)\cdot \big(\bbD(x)\nabla g(x)\big)=
\Big(\frac{2-\beta}{2}\Big)^2\delta(x)^{\beta-2} \Big[\Big(1-\frac{1}{2 \ln \frac{1}{\delta(x)}}\Big)^2
+ c_{iso}(\nu_0)\delta^{1-s}(x)\Big].
\end{equation}
Noticing that for $\delta(x)$ sufficiently small
\begin{equation}
\Big(1-\frac{1}{2 \ln \frac{1}{\delta(x)}}\Big)^2
+ c_{iso}(\nu_0)\delta^{1-s}(x) \leq 1
\end{equation}
from this point the proof repeats the proof of Theorem \ref{Mbetaconst}. 
Since $ V=q=0,\; \rho=1$, \eqref{E:ESA} gives the condition 
$\Big(\frac {\beta-1}{\beta -2}\Big)^2 \geq 1$ hence $\beta \geq \frac{3}{2}$.
\end{proof}

\begin{remark}

A similar argument leads to the fact that if \eqref{C:Disot} holds true then Theorem \ref{Mbetaconst} implies
Theorem \ref{Ro} i. i.e., $H_0$ is essentially self-adjoint for $\beta > \frac{3}{2}$.
\end{remark}

At the level of power behaviour and for $\bbD$ asimptotically isotropic as $x \rightarrow \partial \Omega$ the problem of dependence of self-adjointness of $H_0$ on the degeneracy of $\bbD$
 is completely solved by Theorem \ref{Ro} and Proposition \ref{PR}. The transition from essential self-adjoint to not essential self-adjoint takes place at the critical  value $\beta = \frac{3}{2} $. It is an interesting problem to "zoom" the neighbourhood of the critical value and to obtain a more precise description of the essential self-adjoint to not essential  self-adjoint transition. In particular, assuming that as $x \rightarrow \partial \Omega$, $\bbD(x) \sim \delta(x)^{\frac{3}{2}}\big(\ln\frac{1}{\delta(x)}\big)$ the problem is to find the critical value, $\alpha_c$, of $\alpha$ where the transition essential self-adjoint to not essential self-adjoint takes place. Notice that from Proposition \ref{PR} $\alpha =0$ lies in the essential self-adjoint region.

\begin{theorem}\label{C:ESAlog}
Let  $\Omega$ is bounded, simply connected and with $\cC^2$ boundary. Suppose
that $\bbD \in \cC^1(\Omega; \bbR^{d \times d})$ and there exists $r_0 >0,\; s<1$ such that on 
$\Gamma_{r_0}$
 
 \begin{equation}\label{C:Disots}
  \|\bbD(x)\delta(x)^{-\frac{3}{2}}\Big(\ln\frac{1}{\delta(x)}\big)^{-\alpha}- d(x)\mathds{1}\| \leq c_{r_0} \delta(x)^{1-s} 
\end{equation}
with $d \in \cC^1(\Omega)$ and $d,\; \frac{1}{d},\;
|\nabla d| \in L^{\infty}\big(\Gamma_{r_0}\big)$.
Then $H_0$ is essentially self-adjoint for $\alpha \leq \frac{1}{4}$.
 \end{theorem}
 
 \begin{proof}
Proceeding  as in the proof of Proposition \ref{PR} we obtain from Theorem \ref{C:Hln} that for $\nu$ sufficiently small
 \begin{equation}\label{H:Hardyalpha32} 
h_0\big[\varphi, \varphi\big]
\geq -c_1(\nu) \|\chi_{\Omega \setminus \Gamma{\nu}}\varphi \|^2+ \int_{\Gamma_{\nu}}\cH_{\alpha}(x)|\varphi(x)|^2 \rho(x) dx
\end{equation}
with
\begin{align}\label{H:calHalpha32}
&\cH_{\alpha}(x)= \nonumber \\
&\frac{1}{4}d(x)\delta(x)^{-\frac{1}{2}} \Big(\ln\frac{1}{\delta(x)}\Big)^{\alpha}\Big[ \big(\frac{1}{2} -\frac{\alpha}{\ln\frac{1}{\delta(x)}}\big)^2
+\frac{1-2\alpha -\alpha^2}{\big(\ln\frac{1}{\delta(x)}\big)^2}\Big]=\nonumber \\
&\frac{1}{16}d(x)\delta(x)^{-\frac{1}{2}} 
\Big(\ln\frac{1}{\delta(x)}\Big)^{\alpha}\Big(1-
\frac{4\alpha}{\ln \frac{1}{\delta(x)}}+
\frac{4(1-2\alpha)}{\Big(\ln\frac{1}{\delta(x)}\Big)^2}\Big).
\end{align}
Choosing $g(x)$ as in Proposition \ref{PR} one obtains in this case
\begin{align}\label{E:gradgMisoalpha}
&\nabla g(x)\cdot \big(\bbD(x)\nabla g(x)\big)=\nonumber \\
&\frac{1}{16}\delta(x)^{-\frac{1}{2}}\Big(\ln\frac{1}{\delta(x)}\Big)^{\alpha} \Big[\Big(1-\frac{1}{2 \ln \frac{1}{\delta(x)}}\Big)^2
+ c_{iso}(\nu_0)\delta^{1-s}(x)\Big].
\end{align}

Then as in Theorem \ref{Mbetaconst} the essential self-adjointness of  $H_0$ for $\alpha \leq \frac{1}{4}$ follows from the fact that for sufficient small $\delta(x)$
 \begin{equation}
 \frac{1-
\frac{4\alpha}{\ln \frac{1}{\delta(x)}}+
\frac{4(1-2\alpha)}{\Big(\ln\frac{1}{\delta(x)}\Big)^2}\Big)}{\Big(1-\frac{1}{2 \ln \frac{1}{\delta(x)}}\Big)^2
+ c_{iso}(\nu_0)\delta^{1-s}(x)} \geq 1.
 \end{equation}
 
 \end{proof}
8. During the last decade U. Boscain and his collaborators made a detailed study of essential self-adjointness (aka quantum confinement) of Laplace-Beltrami operator on two-dimensional manifolds endowed with almost Riemannian structures (2-ARS) i.e. Riemannian structures which become singular on some one-dimensional submanifolds; we send the reader to the recent paper \cite{BBP} for a detailed discussion and further references. 

Since, as noticed in \cite{NN2}, the drift-diffussion operator can be viewed (up to a sign) as Laplace-Beltrami operator on weighted Riemannian manifolds it is a natural questions whether the results on 
drift-diffusion operator are related to the  results for Laplace-Beltrami operators 2-ARS. Let us notice first that the result in \cite{NN2} as well as those in \cite{Ro1}, \cite{Ro2} do not apply to the setting in \cite{BBP} since they require $\bbD(x)$ to be in a neighbourhood of the isotropic case (see Comment 6 above). On the contrary since the results in this paper do not require
$\bbD(x)$ to be almost isotropic Corollary \ref{C:2ARS} below can applied to the crucial step in proving the essential self-adjointness
 of Laplace-Beltrami operator on 2-ARS \cite{BL}, \cite{BBP}, namely the analysis in a tubular neighbourhood of the singular set.

Let $\Omega$ be the unit disc in $\bbR^2$, $\Omega =\big\{x=(x_1,x_2)\big| |x| <1\big\}$ and for $0< \nu_0 <1 $,

$$\Gamma_{\nu_0}=\big\{ x\in \Omega \big| \big(1-|x| \big)< \nu_0\big\}.$$

Consider on $\Omega$ a Riemannian distance such that on $\Gamma_{\nu_0}$:
\begin{equation}
ds^2= dr^2 + \frac{1}{(1-r)^{2\alpha}e^{2\Phi(r,\theta)}}d\theta^2,
\end{equation}
where $\alpha \in \bbR; \;r, \theta$ are the standard polar coordinates in $\Omega$ and  $\Phi \in 
\cC^{\infty}\big( \Gamma_{\nu_0}; \bbR\big)$ with uniformly bounded second order derivatives.
On $\Gamma_{\nu_0}$ the Riemannian volume $\omega_{\alpha}$, Laplace-Beltrami operator $\Delta_{\alpha}$ and the
curvature $K_{\alpha}$ have the form:
\begin{equation}
\omega_{\alpha} =\frac{1}{f_{\alpha}(r, \theta)} dr d\theta,
\end{equation}

\begin{equation}
\Delta_{\alpha}= \partial_r^2+ f_{\alpha}^2 \partial_{\theta}^2 +\frac{ \big(\partial_r f_{\alpha}\big)}{f_{\alpha}} \partial_r
+ f_{\alpha}\big(\partial_r f_{\alpha}\big)\partial_{\theta},
\end{equation}

\begin{equation}
K_{\alpha}(r,\theta)=\frac{f_{\alpha}(r,\theta)\big(\partial_r^2 f_{\alpha}(r, \theta)\big)-4\big(\partial_r f_{\alpha}(r, \theta)\big)^2}{f_{\alpha}(r, \theta)^2}
\end{equation}

where 

\begin{equation}
f_{\alpha}(r, \theta)\big) =(1-r)^{\alpha}e^{\Phi(r,\theta)}.
\end{equation}

Consider now the operator \cite{BBP}
\begin{equation}
H_{\alpha,c}=-\Big(\frac{1}{2}\Delta_{\alpha} +c K_{\alpha}\Big); \quad c \in\bbR;
\quad \cD(H_{\alpha,c})=\cC_0^{\infty}(\Omega).
\end{equation}

\begin{corollary}\label{C:2ARS}
Suppose that
\begin{equation}\label{C:HARS}
\Big(\frac{1+\alpha}{2}\Big)^2 -2c\alpha(1+\alpha) \geq 1.
\end{equation}
Then $H_{\alpha,c}$ is essentially self-adjoint in $L^2\big(\Omega;  \omega_{\alpha}(x)dx\big)$.
\end{corollary}

\begin{proof}
By direct computation, on $\Gamma_{\nu_0}$:
\begin{equation}\label{C:ARSK}
K_{\alpha}(r,\theta)= -\frac{\alpha(\alpha +1)}{(1-r)^2}\big(1 + W_{\alpha}(r, \theta)\big)
\end{equation}
with
\begin{equation}
\sup_{x \in \Gamma_{\nu_0}}\big|W_{\alpha}(r, \theta)\big| (1-r)= w(\alpha) <\infty.
\end{equation}

On $\Gamma_{\nu_0}$, with the identifications:

\begin{equation}\label{C:ARSrho}
\rho(x)=\frac{1}{(1-r)^{\alpha}re^{\Phi(r,\theta)}},
\end{equation}

\begin{equation}\label{C:ARSD}
\bbD(x)=\frac{1}{2}\Big(\bbP_{\mathbf e_r} + (1-r)^{2\alpha}r^2e^{2\Phi(r,\theta)}
\bbP_{\mathbf e_{\theta}}\Big)
\end{equation}

where $\bbP_{\mathbf e_r},\; \bbP_{\mathbf e_{\theta}}$ are the orthogonal projections in $\bbR^2$ along $\mathbf e_r =(\cos \theta, \sin \theta),\; \mathbf e_{\theta}=
(-\sin \theta, \cos \theta)$ respectively and
\begin{equation}\label{C:ARSV}
V(x)=K_{\alpha}(r,\theta),
\end{equation}
$H_{\alpha,c}$ has the form \eqref{S:H}.

From \eqref{C:ARSrho}-\eqref{C:ARSV} $\rho,\; \bbD, \; V$ satisfy Asssumption A and
\eqref{E:Vbetaconst} (notice that $\delta(x) =(1-r)$) with $\beta=0,\; \gamma(x)=-\alpha,\;
a(x)=\frac{1}{2},\; v(x)= -c\alpha(\alpha+1), s=0$ and $q(x)= 0$. Then Corollary \ref{C:2ARS}
follows from Theorem \ref{Mbetaconst}ii. 
\end{proof}

Particular  cases  of  Corrolary \ref{C:2ARS} are related to the sufficient conditions for esential
self-adjointnes in \cite{BL}, \cite{BBP}; in the case $\alpha =1$, $\Phi (1,  \theta)=0$ it
gives   another proof of Theorem 1 in \cite{BL} while the case $\Phi (r, \theta) = 0$
corresponds to the sufficient part of Proposition 4 in \cite{BBP}. One may speculate  that the results in Section \ref{S:4} can provide sufficient conditions ensuring essential self-adjointness of
Laplace-Beltrami operator on larger classes of ARS in $d\geq 2$.

\end{document}